\documentclass[12pt, draftclsnofoot, onecolumn]{IEEEtran}
\usepackage{amsmath}
\usepackage{amssymb}
\usepackage{amsfonts}
\usepackage{dsfont}
\usepackage{soul}
\usepackage{color}
\usepackage{bm}
\usepackage{graphicx}
\usepackage{float}
\usepackage{subfigure}
\usepackage{amsthm}
\usepackage{url}
\usepackage{cite}
\newtheorem{theorem}{Theorem}
\newtheorem{lemma}{Lemma}
\newtheorem{remark}{Remark}
\newtheorem{assumption}{Assumption}

\begin{document}

\title{\huge Deploying Federated Learning in Large-Scale Cellular Networks: Spatial Convergence Analysis}

\author{Zhenyi Lin, Xiaoyang Li, Vincent K. N. Lau, Yi Gong, and Kaibin Huang
\thanks{Z. Lin and K. Huang are  affiliated with The University of Hong Kong, Hong Kong. Z. Lin is also with the Dept. of EEE at Southern University of Science and Technology (SUSTech), China. X. Li and Y. Gong are with the same institute. V.  K. N. Lau is with the The Hong Kong University of Science and Technology, Hong Kong. Corresponding authors: K. Huang (Email: huangkb@eee.hku.hk), Y. Gong (Email: gongy@sustech.edu.cn)}
}

\maketitle

\IEEEpeerreviewmaketitle
\begin{abstract}
The deployment of federated learning in a wireless network, called \emph{federated edge learning} (FEEL), exploits low-latency access to distributed mobile data to efficiently  train an AI model  while preserving data privacy. In this work, we study  the spatial (i.e., spatially averaged) learning performance of FEEL deployed in a large-scale cellular network with spatially  random distributed devices.  Both the schemes of digital and analog transmission are considered, providing support of error-free uploading  and over-the-air aggregation  of local model updates by devices. The derived spatial convergence rate for digital transmission is found to be constrained  by a limited  number of active devices regardless of device density and converges to the ground-true rate exponentially fast as the number grows. The population of active devices depends on network parameters such as processing gain and signal-to-interference threshold for decoding. On the other hand, the limit does not exist for uncoded analog transmission. In this case, the spatial convergence rate is slowed down due to the direct exposure of signals to the perturbation of inter-cell interference. Nevertheless, the effect diminishes when devices are dense as interference is  averaged out by aggressive over-the-air aggregation. In terms  of learning latency (in second), analog transmission is preferred to the digital scheme as the former dramatically reduces multi-access latency by enabling simultaneous access. 
\end{abstract}

\section{Introduction}
The availability of enormous data at edge devices motivate the deployment of machine-learning algorithms at the network edge to distill the data into \emph{artificial intelligence} (AI). The trained AI models are expected to enable a wide range of  next-generation mobile applications such as  autonomous driving and  augmented reality. Fast growing relevant research has led to the emergence of a new area called \emph{edge learning}  \cite{8736011,zhu2020toward}. In this area, federated learning is perhaps the most widely studied framework due to its feature of preserving data privacy by avoiding their uploading. To this end, a model-training task  is distributed over devices using the iterative algorithm of \emph{stochastic-gradient descent} (SGD) \cite{lim2020federated}.  A main vein of research on edge learning concerns efficient implementation of federated learning in wireless systems, call \emph{federated edge learning} (FEEL). In this work, we study the performance of FEEL in a large-scale cellular network where inter-cell interference is present. The results help crystalizing the effects of network parameters on the (model) convergence rate.  

In the area of FEEL, recent years have seen the development of diversified approaches for overcoming the communication bottleneck, which is caused by the uploading of high-dimensional model updates from multiple devices to a server. One approach is efficient joint management of communication-and-computation resources via designing scheduling and bandwidth allocation to accelerate convergence \cite{chen2019joint, shi2020joint, ren2020scheduling, yang2020delay}. 
From the theoretic perspective, researchers have attempted to shed light on the fundamental question of how many devices are needed for providing a guarantee on learning performance  within a finite time duration \cite{song2020optimal}. An alternative approach is to realize ``over-the-air aggregation" of  local model updates so as to support simultaneous access by many devices \cite{zhu2020over, amiri2020machine}. The core idea is to adopt analog transmission so as to exploit the waveform-superposition property of a multi-access channel. The versatility and efficiency of over-the-air aggregation has been improved by the development of numerous  relevant techniques including digital aggregation  \cite{zhu2020one},  gradient compression \cite{amiri2020machine}, power control \cite{zhang2020gradient}, and beamforming  \cite{yang2020federated}. Another approach is source compression. Some existing techniques exploit  local-model sparsity  \cite{lin2017deep} or enable efficient model quanization \cite{amiri2020federated}. 

In view of prior work, most results assume  single-cell systems. The topic of deploying  FEEL in a large-scale network remains one largely unexplored. In this scenario, learning performance is affected by inter-cell interference as well as network configurations. Recently, some initial work has accounted for such an effect in designing device-scheduling schemes  \cite{yang2019scheduling}. While the work points to the  important direction of FEEL networking, many fundamental questions remain unanswered. In particular, a question of our interest is how the convergence depends on the network parameters (i.e., device density, cell sizes, and coding rates),  which parameterize  the interference distribution. 

A standard  approach of characterizing the effect of   inter-cell interference on network performance, which is also adopted in this work, is to model the randomly located network nodes (devices or base stations) as spatial point processes such as  a \emph{Poisson point process} (PPP) or its derivatives \cite{haenggi2009stochastic}. Consequently, the interference power can be modelled as a shot-noise process, referring to a sum over a PPP \cite{kingman2005p}. Then the study of network performance  is reduced to the equivalent  analysis of the  expected performance of a \emph{typical cell}, which results from  uniformly  sampling all cells, over the distributions of interference, channels, and nodes   \cite{haenggi2009stochastic}. Such analysis  leverages   a rich set of results  from the stochastic-geometry theory \cite{chiu2013stochastic}. The tractability brought by the theory has motivated many researchers to use it as a tool to study the performance of a wide range of wireless networks such as cellular networks  (see e.g., \cite{andrews2011tractable}), cooperative networks (see e.g., \cite{hosseini2016stochastic}), heterogeneous networks (see e.g., \cite{dhillon2012modeling, soh2013energy}), and most recently unmanned aerial vehicle networks \cite{chetlur2017downlink}. Most existing work is based on the classic ``communication-and-computation separation" approach. To be specific, the considered networks aim at providing generic  radio-access services to users or sensors without concerning their  applications. The corresponding design objective is to  ensure the required quality-of-service, network throughput  or   coverage \cite{haenggi2009stochastic}. In contrast, the study of a FEEL network, referring to a network supporting the FEEL application, should adopt a learning-related  metric for network performance such as the proposed metric of convergence rate in a typical cell, termed  \emph{spatial convergence rate}. The corresponding network-performance analysis is differentiated from existing analysis in  its  interplay of stochastic geometry  and learning theories, which is a key feature of current analysis. 

In this work, we consider a large-scale network with hexagonal cells and devices distributed following a PPP. FEEL is deployed in a typical cell.  For the reason, the corresponding model convergence is termed \emph{spatial convergence}. Uplink transmission by each device is based on either digital or analog (over-the-air aggregation) transmission and protected against interference using \emph{frequency-hopping spread spectrum} (FHSS) following \cite{1542405}. By analyzing the spatial convergence  rate, we quantify the effects of network parameters on the learning performance for different transmission schemes and scenarios (i.e., low and high mobility). The key findings are summarized as follows.

\begin{itemize}
    \item\textbf{Spatial convergence for digital transmission:} The spatial convergence rate (in terms of rounds) \cite{bernstein2018signsgd} is derived to quantify the  deviation from the ground-true rate, which corresponds to direct gradient descent on the loss function. The deviation results from inter-cell interference and a random number of devices that succeed in transmission (i.e., a random data size), called \emph{successful devices}. The key findings are as follows. First, as the device density grows, the expected number of successful devices converges to a constant  and thereby introduces a limit to the learning performance. The expected number is proportional to the processing gain of spread spectrum, decreases with a growing \emph{signal-to-interference} (SIR) threshold for successful transmission, but is insensitive to variations of cell sizes.  Second, the mentioned rate deviation diminishes \emph{exponentially fast} as the expected number of successful devices increases.  Last, channel-temporal diversity due to high mobility increases the chance of  a device to succeed in transmission and participate in at least  one round of the learning process, which increases   the spatial-convergence rate. 
    
    \item \textbf{Spatial convergence for analog transmission:} The distinctions of analog transmission is its support of simultaneous access while directly exposing  the received  model update to the perturbation by interference. The corresponding spatial convergence rate is derived by applying results on the interference distribution from stochastic geometry to the convergence analysis. The rate deviation from the ground truth reveals two conflicting effects of increasing the devices density. On one hand, without outage, the expected number of active devices participating in learning can grow unboundedly as the density increases. Consequently, more training data lead to faster spatial convergence. On the other hand, increasing the device density also  causes the number of significant interferers to grow, which perturbs the SGD process and slows down spatial convergence.  As the first scaling law is faster than the second, the net effect is found to be a higher spatial convergence rate when devices are denser. This makes analog transmission a favourable choice over the digital counterpart in a dense network. 
    
    \item \textbf{Learning Latency:} Besides corroborating  the above findings, experiments using  a real dataset are conducted to compare the learning latency (in second) of digital and analog transmission. The latency  of  analog transmission is observed to be much lower than the  digital-transmission counterpart  in both sparse and dense networks. The low-latency of analog transmission  in a sparse network results from more active devices (i.e., fewer rounds) and that in a dense network  from shorter per-round latency. The findings are aligned with those for a single-cell system \cite{zhu2020over}.

\end{itemize}

The remainder of this paper is organized as follows. Models and metrics are introduced in Section II. Spatial convergence is analyzed \emph{with respect to} (w.r.t.) for the cases of digital-transmission and analog transmission in  Sections III and IV, respectively. Experimental results are presented in Section V, followed by concluding remarks in Section VI.

\section{Models and Metrics}

\subsection{Network Topology Model}

Adopting the  classic model, the cellular network contains hexagonal cells as illustrated  in Fig.~\ref{FigSys} \cite{Rappaport1996Wireless}. Base stations, denoted as  $\{Y\}\subset\mathds{R}^2$, are placed at cell centers. Let $C(Y, R)$ denote a cell centered at $Y$ and with a distance $R$ from $Y$ to its boundary. Randomly located edge devices, denoted as  $\{X\}\subset\mathds{R}^2$, are  randomly distributed on  plane modeled as a homogeneous  PPP $\Phi_\text{d} = \{X\}$ with density $\lambda_{\text{d}}$.  FEEL is performed in a \emph{typical cell} chosen by uniform sampling of all cells, denoted as $C_0 = C(Y_0, R)$ with $Y_0$ being the typical BS \cite{haenggi2009stochastic}.  Devices in other cells are interferers involved in  other services or tasks. To facilitate analysis, the number of devices in $\mathcal{C}_0$, namely $|\mathcal{C}_0\cap\Phi_\text{d}|$,  can be lower bounded by $K = |\mathcal{D}_0\cap\Phi_\text{d}|$, where  $\mathcal{D}_0$ represents  the inscribed disk of $\mathcal{C}_0$ with the radius $R$ (see Fig. \ref{FigSys}). For the $K$ devices, their propagation distances to $Y_0$ are \emph{independent and identically distributed} (i.i.d.) with  the following  \textit{probability density function} (PDF):
\begin{eqnarray}
    f_R(r) = \frac{2r}{R^2},~~~ 0<r<R.
\end{eqnarray}

\begin{remark}[Extension to Random Cells]\emph{It is possible to extend the current results to the case of random cells generated by  BSs distributed as a Poisson point process instead of a hexagonal lattice \cite{andrews2011tractable}. Similar to the current case, a random typical cell can be inner bounded by a disk but is radius, $R$,  is now random. Specifically, $R$ has the distribution function of $f_R(r) = 8\pi\lambda_{\text{s}} r \exp(-4\lambda_{\text{s}}\pi r^2)$ \cite{andrews2011tractable}.  The current analytical results hold conditioned on a given $R$. Then taking their expectation with respect to the distribution of $R$ yields the desired extension. }
\end{remark}

\begin{figure*}[t]
\centering
\includegraphics[width=9cm]{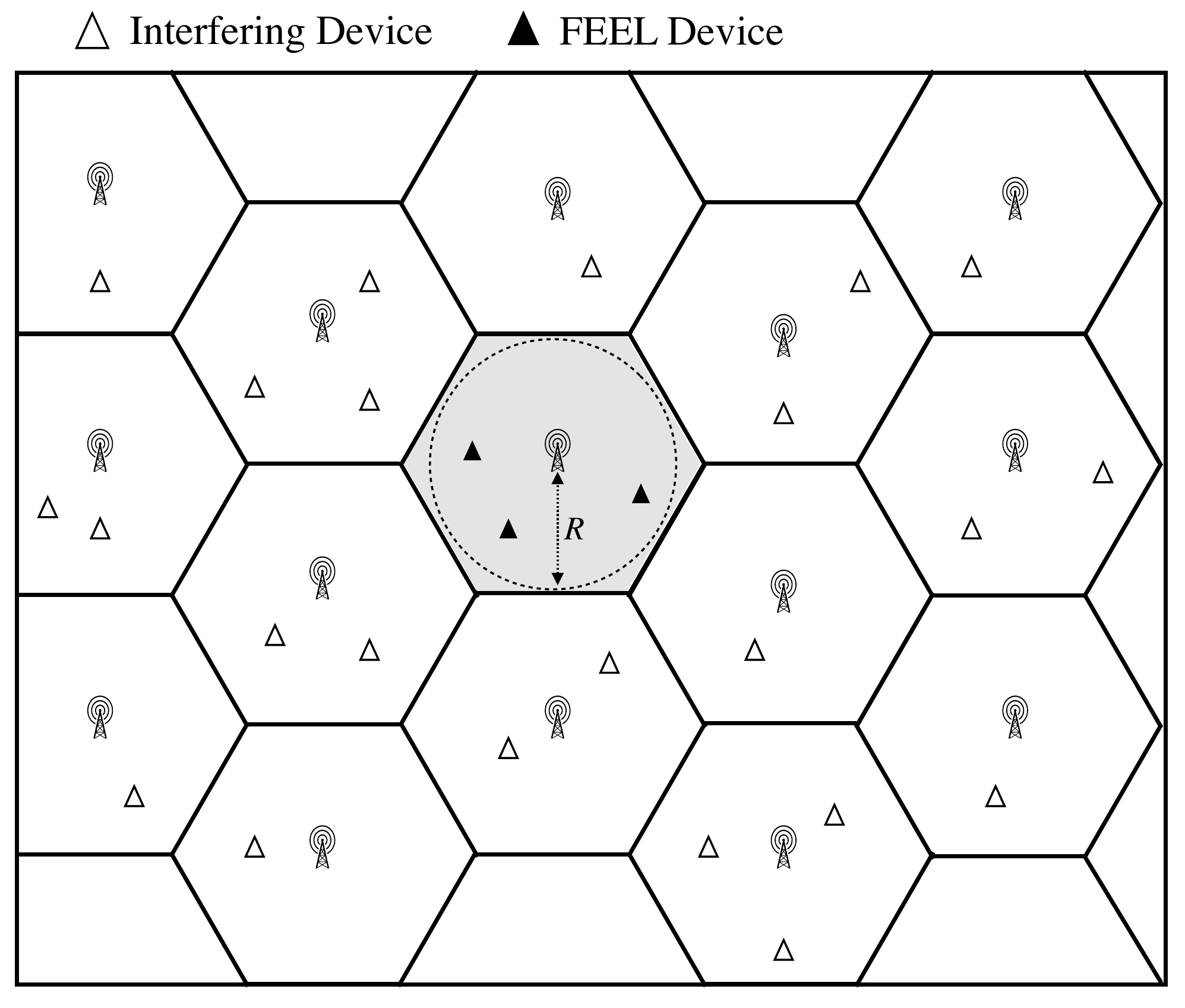}
\caption{The spatial model of a cellular network where  FEEL is supported in a typical cell.}
\label{FigSys}
\end{figure*}

\subsection{Federated Learning Model}

The operations of FEEL is illustrated in Fig.~\ref{FigLearn} and described as follows. We consider the specific implementation of FEEL where stochastic gradients are computed at devices using local data and then transmitted to the server (co-located with the BS) for updating the global model \cite{lim2020federated} (see Remark~\ref{Re:LocalModelExt} for extension to alternative implementation). Each round of FEEL comprises three phases: (1) global model updating and broadcasting, (2) local gradient computation, and (3) local gradient uploading. The current analysis focuses on  the last phase as it represents the communication bottleneck of the FEEL system as discussed earlier.  Let $t^{(n)}_{\text{cmm}}$ denote the duration of the uploading phase in the typical cell  in  the $n$-th round. The requirement that all participating devices must finish their uploading within the duration before the global model can be updated introduces the constraint of so called synchronized updates \cite{zhu2019broadband}.  Under the constraint, $t^{(n)}_{\text{cmm}}$ is a random variable depending on the random number of workers in the cell and their channel states. In contrast, the broadcasting phase uses the whole spectrum and can be assumed to finish within a given duration denoted as $t_{\text{bc}}$. Moreover, the workers are assumed to have comparable computation capacities, enabling them to complete local computation within a given duration denoted as $t_{\text{cmp}}$.

\begin{figure*}[t]
\centering
\includegraphics[scale=1]{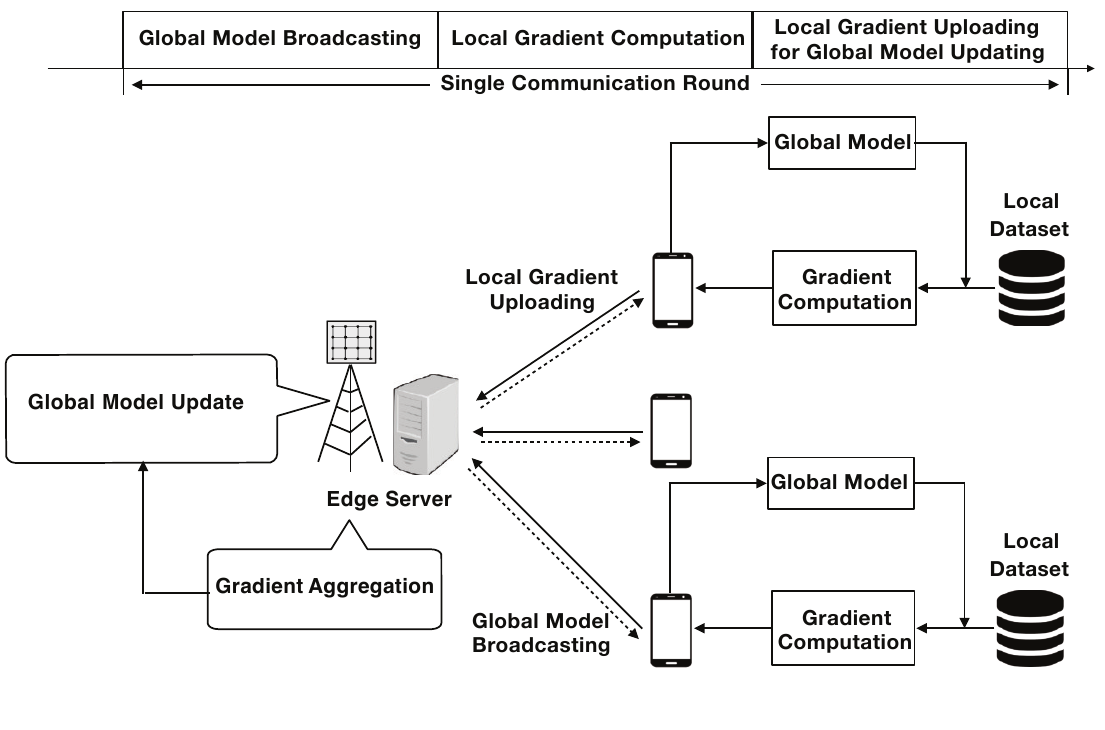}
\caption{The operations of FEEL in the  typical cell.}
\label{FigLearn}
\end{figure*}

Let $N$ denote the number of rounds needed for accomplishing the learning task, and $A^{(n)}$ the number of devices that successfully upload their gradients in the $n$-th round. Let $I_{X}$ be an indicator function of worker ${X}$ with $I_{X} = 1$ if transmission is successful  or otherwise $I_{X} = 0$. Denote the device process  in the $n$-th round  as $\Phi_\text{d}^{(n)}$.  Then we can write $A^{(n)} = \sum\nolimits_{{X}\in \mathcal{C}_0\cap\Phi_\text{d}^{(n)}} I_{X}$. Both the cases of  high and low mobility are considered. In the case of high mobility,  $\{\Phi_\text{d}^{(n)}\}$ are  independent over different rounds and so are $\{A^{(n)}\}$. In the case of low mobility, they are fixed throughout the learning process: $\Phi_\text{d}^{(1)} = \Phi_\text{d}^{(2)} = \cdots = \Phi_\text{d}^{(N)}$ and thus $A^{(1)} = A^{(2)} = \cdots = A^{(N)}$.

In the current setting of supervised learning, let a labelled data sample be denoted as $(\bm{u},y)$ with $\bm{u}$ and $y$ representing the data and label, respectively. The samples follow an \emph{unknown} probability distribution $p(\bm{u},y)$. Let $\bm{w}$ denote the model or its parameters.  Consider the loss function $f(\bm{w}; \mathbf{u}, y)$, which measures the discrepancy between predicted output from $\bm{w}$ using  the sample $(\mathbf{u}, y)$. The expected risk of the predictor $\bm{w}$, known as the \emph{ground-true loss function}, is defined as \cite{konevcny2017stochastic}:
\begin{equation}
    F(\bm{w}) = \mathsf{E}_{(\mathbf{u},y)\sim p(\mathbf{u},y)}[ f(\bm{w}; \mathbf{u},y)].
    \label{global loss f}
\end{equation}
Since the data distribution $p(\mathbf{u},y)$ is unknown, it is impossible to find the ideal model $\bm{w}^*=\arg\min_{\bm{w}} F(\bm{w})$. FEEL is a distributed  training algorithm for finding an approximate of the ideal model, which is described as follows.  

To this end, some notation is introduced. The local dataset of device ${X}$ is denoted as $\mathcal{D}_{X}$ comprising samples that are drawn i.i.d. from $p(\mathbf{u},y)$. Then the \emph{local loss function} is defined in terms of the \emph{empirical risk} as \cite{zhu2020one}:
\begin{equation}
    F_{X}^{(n)}(\bm{w}^{(n)}) = \frac{1}{|\mathcal{D}_{X}|} \sum_{(\mathbf{u}, y) \in \mathcal{D}_{X}} f(\bm{w}^{(n)}; \mathbf{u}, y). 
\label{local_loss}
\end{equation}
For convenience, we assume a uniform size for local datasets, i.e., $|\mathcal{D}_{X}|\equiv D,~\forall X$. The learning task of the typical cell is specified by the tuple $\{F_0, F^*, f\}$, where $f$ is the mentioned  per-sample loss function, $F_0 \triangleq F(\bm{w}^{(0)})$ denotes the value of the ground-true loss function $F$ at  the initial model $\bm{w}^{(0)}$, and $F^*$ is the global minimum of $F$.

The distributed SGD algorithm underpinning FEEL is described as follows  (see e.g., \cite{mcmahan2017communication}). Consider the  $n$-th  round, each device uses its local dataset $\mathcal{D}_{X}$ and the model broadcast by the BS,  $\bm{w}^{(n)}$,  to compute  the gradient of the local loss function $F_{X}^{(n)}(\bm{w}^{(n)})$, called a local gradient and denoted as $\tilde{\bm{g}}_{X}^{(n)}=\nabla F_{X}^{(n)}(\bm{w}^{(n)})$.
The local gradients are transmitted  to the BS for averaging, yielding  the following \emph{global gradient estimate (of that of the ground-true loss function)}:
\begin{equation}
    \bar{\bm{g}}^{(n)}_0=\left\{
        \begin{aligned}
            & \frac{1}{A^{(n)}} \sum\limits_{{X}\in \mathcal{C}_0\cap\Phi_\text{d}^{(n)}} \tilde{\bm{g}}_{X}^{(n)}, & A^{(n)}\geqslant 1, \\
            & 0, & A^{(n)} = 0.
        \end{aligned}
    \right.
\end{equation}
It is applied to updating the global model based on  gradient descent: 
\begin{equation}
\bm{w}^{(n+1)} = \bm{w}^{(n)} - \mu\bar{\bm{g}}_0^{(n)}, 
\label{Eq:global_update}
\end{equation}
where the step size $\mu$ is called  the \emph{learning rate}.  Last, the BS broadcasts the updated model to all devices, completing one round. The rounds are repeated till the model converges. 

\begin{remark}[Extension to Local-model Uploading]\label{Re:LocalModelExt} \emph{The current analysis can be extended to the alternative FEEL implementation with local-model uploading by accounting for multi-round local-gradient descent \cite{lim2020federated}. First, in each round, the local model at device $X$ is updated via $\tilde{\bm{w}}_{X}^{(n+1)} = \bm{w}_{X}^{(n)} -\mu \tilde{\bm{g}}_{X}^{(n)}$; then  $\tilde{\bm{w}_{X}}^{(n+1)}$ is transmitted   to the server for updating the global model:  $\bm{w}_{X}^{(n+1)} =  \frac{1}{A^{(n)}} \sum\nolimits_{{X}\in \mathcal{C}_0\cap\Phi_\text{d}^{(n)}} \tilde{\bm{w}}^{(n+1)}_{X}$. The analysis can be modified accordingly and the modification is   straightforward and  does not change the  findings. 
}
\end{remark}

For tractable convergence analysis, a set of assumptions commonly made in the literature (see e.g.,  \cite{bernstein2018signsgd}) are also adopted in this work. 
\begin{assumption} \emph{ (Lower Bound) The ground-true loss function $F(\bm{w})$ is lower bounded, namely  $F(\bm{w})\geq F^*$ for some constant $F^*$.}
\end{assumption}

\begin{assumption}\emph{(Smoothness) Let $S$ denote the model dimension and hence we can write the parameter vector as $\bm{w} = [w_1, w_2, ..., w_S]^T$. The ground-true loss function $F(\bm{w})$ is assumed smooth. Mathematically, for the loss function evaluated at $\bm{w}$, we assume there exist a non-negative constant vector $\bm{L} := [L_1, L_2, ..., L_S]^T$,  the gradient of the ground-true loss function $F(\bm{w})$, $\nabla F(\bm{w})$, satisfies the following 
\begin{equation}
\left|F(\bm{\beta}) - [F(\bm{w})+\nabla F(\bm{w})^T(\bm{\beta}-\bm{w})])\right| \leq \frac{1}{2}\sum_{i=1}^S L_i(\beta_i-w_i)^2,~ \forall \bm{w},\bm{\beta}.
\label{Eq:ass2}
\end{equation}
Define the   $\mathcal{L}_2$ Lipschitz constant $L_0$ as $L_0 : = \|L\|_{\infty}= \max_i L_i$.}
\end{assumption}
  
\begin{assumption}\emph{(Variance Bound) The stochastic gradient (or local gradient estimate)  $\tilde{\bm{g}}_{X}$ at an arbitrary device, say $X$, is an unbiased estimation of ground-true loss function and has a bounded variance: 
\begin{equation}
\mathsf{E}[\tilde{\bm{g}}_{X}]=\nabla F \quad \text { and } \quad \mathsf{E}\left[||\tilde{\bm{g}}_{X}-\nabla F||^2 \right] \leq \sigma^2,
\end{equation}
where $\sigma^2$ is a given constant. }
\end{assumption}

\subsection{Uplink Transmission Models}
The  cellular  network is assumed to be interference-limited, where channel noise is  negligible. All  BSs and devices are equipped with single antennas. The transmission power of a device in the typical cell depends on  the specific transmission schemes as elaborated in the sequel. For simplicity, all interfering devices are assumed to fix their power as $P$. FHSS is adopted to regulate inter-cell interference \cite{goldsmith2005wireless}. Specifically, the total uplink bandwidth $B$ is divided into $M$ sub-carriers; each device randomly chooses one sub-carrier for transmission in  each round and its choices over rounds are independent. As a result, the devices occupying an arbitrary  sub-carrier, say sub-carrier $m$,  is obtained from  $\Phi_\text{d}$ by thinning and thus also a PPP but with density $\lambda_{\text{d}}/M$, denoted as $\hat{\Phi}_{\text{d}, m}$. The transmission by an arbitrary device, $X$, is received at the typical BS with the power $G_X|X - Y_0|^{-\alpha}$ where the coefficient $G_X = \exp(1)$ models Rayleigh fading and $\alpha$ denotes the path-loss exponent. All fading coefficients are assumed independent. We consider two transmission schemes for devices in the typical cell. They are described as follows. 

\subsubsection{Digital Transmission}
For digital transmission, each coefficient of the local gradient at each device is quantized into a sufficiently large number of bits, denoted as $D$, such that the effect of quantization errors on learning performance is negligible. Then the quantized gradient is encoded and transmitted at the  fixed rate  $\frac{B}{M}\log(1+\theta)$ with $\theta$ being a chosen constant.  The fixed rate yields constant communication latency per round given as
\begin{equation}\label{Eq:tau}
    t_{\text{cmm}} = \frac{S DM}{B\log(1+\theta)}.
\end{equation}
To cope with both intra-cell and inter-cell interference, all devices make independent choices of their hopping patterns, each of which refers to  a sequence of choices of sub-carriers over rounds. Moreover, the transmission power of devices are assumed to be fixed and identical  to that of interferers. Considering the $n$-th round, the  receive SIR for transmission by an arbitrary device in the typical cell, denoted as $X_0$, over a chosen subcarrier, say $m$,  can be written as: 
\begin{eqnarray}\label{Eq: sir}
    \mathrm{SIR}^{(n)}_{X_0} = \frac{G_{X_0}|X_0|^{-\alpha}}{\sum_{X\in \hat{\Phi}_{\text{d}, m}^{(n)}\backslash\{X_0\}} G_{X}|X|^{-\alpha}}.
\end{eqnarray}
If the SIR exceeds the threshold $\theta$, the uploaded gradient can be  decoded correctly or otherwise an outage event occurs, resulting in the device being successful or inactive, respectively  \cite{weber2007effect}. Let $\mathcal{A}^{(n)}$ denote the set of active (or successful) devices in the typical cell in the $n$-th round: 
\begin{eqnarray}
    \mathcal{A}^{(n)} = \{X\in \Phi_\text{d}^{(n)}\cap\mathcal{C}_0|\mathrm{SIR}^{(n)}_{X}\geqslant \theta\}.
\end{eqnarray}
Then  $A^{(n)} = |\mathcal{A}^{(n)}|$. 

\subsubsection{Analog Transmission} In each round,  all devices in the typical cell transmit their  local gradients using linear analog modulation without coding and over the same sub-carrier to perform over-the-air aggregation \cite{zhu2019broadband}. In other words, their hopping patterns are identical but independent of  those of interferers. Following the model in \cite{zhu2019broadband}, assuming i.i.d. data over devices, the distribution  of the local-coefficients at each device is  assumed to have the  mean and variance, denoted as $\nu$ and $\tilde{\sigma}^2$, respectively, which are identical for all devices and known by them. To facilitate power control, a local-gradient vector at each device, say $X$, is normalized before transmission to have  zero mean and unit variance, i.e., $\mathbf{s}^{(n)}_X = \frac{\tilde{\bm{g}}_{X}^{(n)}-\nu}{\tilde{\sigma}}$. Then the normalized vector is analog modulated and transmitted as $\sqrt{P_X}\mathbf{s}^{(n)}_X$, where $P_X$ denotes  the transmission power. Next, for the typical BS to receive a desired average of uploaded  local gradients, their corresponding received signals must have aligned in  magnitude, called \emph{magnitude alignment} \cite{zhu2019broadband}. To this end, power control based on truncated channel inversion is applied to suppress channel fading \cite{cao2020optimized}:
\begin{eqnarray}\label{Eq:power_control}
    P_X = \left\{
        \begin{aligned}
            & \frac{\eta}{G_{X}|X|^{-\alpha}}, & G_{X}\geqslant g_\text{th}, \\
            & 0, & \text{otherwise},
        \end{aligned}
        \right.
\end{eqnarray}
where $\eta$ is the magnitude scaling factor of the received signal and  $g_\text{th}$ is a channel truncation threshold chosen to avoid exceeding an average power budget, denoted as  $\bar{P} = \mathsf{E}[P_X]$. For fair performance comparison with digital transmission,  we set the average transmission power to be $\bar{P} = P$ or equivalently $\mathsf{E}[P_X] = P$.  

The two constants $\eta$ and $g_\text{th}$ are set such that the constraint of average transmission power can be satisfied \cite{zhu2020one, cao2020optimized}.   The reason for not factoring path loss into channel  truncation is similar to that for proportional fairness with the ``fairness" measure modified as  data diversity in the current context. In other words, preventing devices with high path loss from transmission would fail to exploit data distributed at the cell edge for learning, and thus scarifies data diversity;  small-scale fading based truncation  in \eqref{Eq:power_control} avoids such an issue. By reuse of notation,  let $\mathcal{A}^{(n)}$  re-denotes the set of devices  whose channels are not truncated in $n$-th round:
\begin{eqnarray}
    \mathcal{A}^{(n)} = \{X\in \Phi_\text{d}^{(n)}\cap\mathcal{C}_0|G_{X}\geqslant g_\text{th}\}.
\end{eqnarray}
Then $A^{(n)} = |\mathcal{A}^{(n)}|$ is the number of active devices in the $n$-th round.

Given analog transmission, the received aggregated signal vector at the typical BS is 
\begin{eqnarray}
    \mathbf{y_0} = \sum\limits_{X\in \mathcal{A}^{(n)}}\sqrt{P_X G_{X}} |X|^{-\frac{\alpha}{2}} \mathbf{s}_{X}^{(n)} + \mathbf{I}_0,
\end{eqnarray}
where $\mathbf{I}_0$ is the  interference given as 
\begin{equation}
    \mathbf{I}_0 = \sum\limits_{X'\in \hat{\Phi}^{(n)}_{d, m}\cap\bar{\mathcal{C}}_{0}} \sqrt{P G_{X'}} |X'|^{-\frac{\alpha}{2}} \mathbf{s}_{X'}^{(n)}.
\end{equation}
At the typical BS, the desired estimation of the aggregated gradient is obtained by the following de-normalization of the received signal \cite{zhu2020one, cao2020optimized}:
\begin{eqnarray}\label{Eq:ana_recv_signal}
    \bar{\bm{g}}_0^{(n)} &=& \frac{\tilde{\sigma}}{A^{(n)}\sqrt{\eta}} \mathbf{y_0} +\nu \nonumber \\
    & = & \frac{1}{A^{(n)}} \sum\limits_{X\in \mathcal{A}^{(n)}} \tilde{\bm{g}}_{X}^{(n)} + \frac{\mathbf{I}_0 \tilde{\sigma}}{A^{(n)} \sqrt{\eta}}.
\end{eqnarray}
Since the  symbol duration is $T_\text{s} = \frac{M}{B}$,  the per-round latency for the analog transmission is 
\begin{eqnarray}\label{Eq:tau_ana}
    t_{\text{cmm}} = \frac{SM}{B}.
\end{eqnarray} 

\subsection{Learning Performance Metrics}
Two metrics for measuring the  performance of FEEL in a spatial network are defined as follows.  The first is the \emph{spatial convergence criterion}.  Consider FEEL in a specific cell centered at a fixed location $y\in \mathds{R}^2$.  Given $N$ rounds, let $\mathcal{J}(N)$ denotes the index set of rounds with a non-empty cell  and the number of \emph{effective rounds} $N_{\text{e}} = |\mathcal{J}(N)| $.  A   convergence criterion  widely adopted in the FEEL literature (see e.g., \cite{bernstein2018signsgd})  is determined by the expectation of averaged-gradient norm over rounds:
\begin{eqnarray}\label{eq:ap_conv}
\bar{g}_0(N) = \mathsf{E}\left.\left[\frac{1}{N_{\text{e}}} \sum_{n \in \mathcal{J}(N)}\left\|\nabla F(\bm{w}^{(n)})\right\|^{2}\right| N_{\text{e}}\geqslant 1\right] \leqslant \varepsilon_0,
\end{eqnarray}
where  $\varepsilon_0$ is a given constant.  Note that the expectation in (\ref{eq:ap_conv}) is taken over the distribution of descent trajectories. Since the typical cell $\mathcal{C}_0$ results from uniform sampling of all cells, there exists a probability that the learning in the cell fails to meet the convergence criterion in \eqref{eq:ap_conv}: $\mathrm{Pr}\left(\bar{g}_0(N) > \varepsilon_0 \right)$. The spatial convergence criterion is defined as one that the  network can support model training within $N$ rounds with a high probability, $(1-\delta)$. Mathematically,\begin{eqnarray}
\mathrm{Pr}\left(\bar{g}_0(N) > \varepsilon_0 \right) \leqslant \delta.
\label{Eq:spa_con_cri}
\end{eqnarray}
It is worth mentioning that if FEEL is performed in all cells, the probability in \eqref{Eq:spa_con_cri} can be interpreted as the percentage of cells where learning fails to be completed in time.

The next performance metric is \emph{expected learning latency} defined as the expected time duration (in second) required  for learning in the typical cell to meet the spatial convergence criterion in \eqref{Eq:spa_con_cri}. Let $N^\star$ denote the smallest number of rounds for meeting the criterion. The expected learning latency is the expected sum of computation-and-communication latency over $N$ rounds:
\begin{equation}\label{Eq:latency}
          \bar{T}_{\Sigma}   = \mathsf{E}\left[\sum_{n}^{N^\star} t_{\text{cmm}}^{(n)}\right] + N^\star (t_{\text{cmp}} + t_{\text{bc}}).
\end{equation}

\section{Spatial Convergence   for  the Digital-Transmission Case}
In this section, we consider the digital-transmission case and study the effects of network parameters on the spatial learning performance. To this end, we derive a sufficient condition for meeting the spatial convergence criterion and analyze the corresponding bound on the minimum expected learning latency. Both the cases of high and low mobility are considered. 

\subsection{Spatial Convergence Analysis with Low-Mobility} \label{sec:conv_dig_low}
Consider FEEL in the typical cell with low mobility. For tractability, the analysis in this section focuses on the case where only the subset of devices lying in the inscribed circle of the cell [see Fig. \ref{FigSys} (a)], $\mathcal{D}_0$, upload local gradients while other devices are silent. As it reduces  training data, the corresponding  convergence rate lower bounds  the counterpart involving all devices.

First, we derive the distribution of the number of active (successful) devices in $\mathcal{D}_0$. To this end, define the \emph{success probability}, denoted as $p_{\text{s}}$, as the probability that an arbitrary device  in $\mathcal{D}_0$ succeeds in transmission. Mathematically, 
\begin{eqnarray}
    p_{\text{s}} = \mathsf{E}_X[\Pr(\mathrm{SIR}^{(n)}_{X}>\theta | X\in\mathcal{D}_0)],
\end{eqnarray}
where $\mathrm{SIR}^{(n)}_{X}$ is given in \eqref{Eq: sir}.  Using the well-known Laplace-transform method (see e.g.,  \cite{andrews2016primer}), the probability can be obtained  as shown  in the following lemma.
\begin{lemma}[Success Probability \cite{andrews2016primer}] \label{lemma: ps_dig} \emph{
The success probability of a typical   device in the disc cell $\mathcal{D}_0$ is given as
    \begin{eqnarray}
        p_{\text{s}} = \frac{1-\exp(-aR^2)}{aR^2},
        \label{Eq:psd}
    \end{eqnarray}
    where
    \begin{eqnarray}\label{Eq:a}
        a = \frac{2\pi\lambda_{\text{d}}\mathcal{B}\left(\frac{2}{\alpha},1-\frac{2}{\alpha}\right)}{\alpha M}\theta^{\frac{2}{\alpha}}, 
    \end{eqnarray}
    with $\mathcal{B}(x,y)$ being the beta function: $\mathcal{B}(x, y)=\int_{0}^{1} t^{x-1}(1-t)^{y-1} d t$. }
\end{lemma}

Let $M_X$ denote an indicator whether device  $X$ is successful or not, i.e., $M_X = \mathrm{I}(\mathrm{SIR}_X\geqslant\theta)$. Thereby, the success  devices form a marked PPP represented by  $\tilde{\Phi}_{\text{d}} = \{X, M_X\}$. By applying the theorem of marked PPP, the density  of $\tilde{\Phi}_{\text{d}}$ is obtained as $\lambda_{\text{d}} p_{\text{s}}$ \cite{kingman2005p}.  Let $K$ denote the number of successful devices in $\mathcal{D}_0$ . 
\begin{lemma}[Distribution of the Number of Successful Devices]\label{lemma:dis_dig_device} \emph{
The distribution  of the number of successful devices, $K$, is given as 
    \begin{eqnarray}\label{Eq: Dig_devices}
        \Pr(K=j)=\frac{\exp \left(-\bar{K}\right)\left(\bar{K}\right)^{j}}{j !},
    \end{eqnarray}
with the mean
    \begin{eqnarray}
        \bar{K} = \pi\lambda_{\text{d}} R^2 p_{\text{s}} =\frac{\alpha M (1-e^{-aR^2})}{2\mathcal{B}\left(\frac{2}{\alpha},1-\frac{2}{\alpha}\right)\theta^{\frac{2}{\alpha}}},
        \label{eq:zeta}
    \end{eqnarray}
and $a$ is given in \eqref{Eq:a}. }
\end{lemma}
\begin{remark}[Finite Active Devices]\label{remark:fin_k}\emph{
It should be emphasized that   as  the device density $\lambda_{\text{d}}$ grows, the expected number of successful devices, $\bar{K}$, does not diverge since $p_{\text{s}}$ decreases due to stronger interference according to Lemma \ref{lemma:dis_dig_device}. As a result, $\bar{K}$ converges to a constant: 
\begin{eqnarray}
    \bar{K}\rightarrow\frac{\alpha M  }{2\mathcal{B}\left(\frac{2}{\alpha},1-\frac{2}{\alpha}\right)\theta^{\frac{2}{\alpha}}}, \qquad \lambda_{\text{d}} \rightarrow \infty.
    \label{Eq:mean_device_den}
\end{eqnarray}
}
\end{remark}

In FEEL, increasing the number of successful devices has the  effect of  increasing the batch-size of training data. This   reduces  the variance of the global gradient estimate. Given Assumption 3, it is straightforward to quantify the reduction as shown in the following lemma. 

\begin{lemma} \label{lemma:Redu_var} \emph{ In the typical cell, the number of successful devices in the $n$-th round,  $K^{(n)}$, reduces the  variance of the global gradient estimate as follows: 
    \begin{equation}
        \mathsf{E} \left[\left\|\frac{1}{ K^{(n)}} \sum\limits_{X\in K^{(n)}} \tilde{\bm{g}}_{X}^{(n)}-\nabla F(\bm{w}^{(n)})\right\|^2  \right]\leqslant \frac{\sigma^2}{K^{(n)}}. 
        \label{Eq:redu_var_bound}
    \end{equation}
}
\end{lemma}

In low-mobility case, the number of successful devices in $\mathcal{D}_0$ is fixed throughout the  learning process: $K^{(1)} = K^{(2)} = \cdots = K^{(N)}=K$. Note that $K$ is a random variable since the typical cell is a random process. With $K$ fixed for a particular  typical-cell realization, the model converge has been analyzed extensively in the literature. Specifically,  the following result on the convergence rate can be derived using the method in \cite{bernstein2018signsgd}. 

\begin{lemma}[Fixed-Cell Convergence with Digital Transmission and Low-mobility \cite{bernstein2018signsgd}] \label{lemma:Sing_cell_conv} \emph{ Consider the case with digital transmission and low mobility. With $K$ fixed, and given  the learning rate $\mu = \frac{1}{L_0\sqrt{N}}$, the expected  averaged-gradient norm is bounded as follows:
    \begin{eqnarray}
        \bar{g}_0(K,N) \leqslant \frac{1}{\sqrt{N}} \left[ (F_0-F^*)+\frac{\sigma^2}{K}\right], \qquad K > 0.
        \label{Eq:dig_conv_single}
    \end{eqnarray}
}
\end{lemma}

Since  $K$ is a random variable, so is  the averaged gradient norm. To facilitate spatial convergence analysis, we  apply the Markov inequality to upper bound the norm as
\begin{eqnarray}\label{Eq:markov}
    \mathrm{Pr}\left(\bar{g}_0(K,N) > \varepsilon_0 \right) &=&     \mathrm{Pr}\left(\bar{g}_0(K,N) > \varepsilon_0 \mid K > 0 \right)(1 - p_{\text{null}}) +p_{\text{null}}  \nonumber\\
    &\leq &\frac{\mathsf{E}[\bar{g}_0(K,N)\mid K > 0]}{\varepsilon_0} (1 - p_{\text{null}}) +p_{\text{null}}, \nonumber 
\end{eqnarray}
where the void probability $p_{\text{null}} = \Pr(K = 0)$. It follows from  Lemma \ref{lemma:dis_dig_device} that $p_{\text{null}}= e^{-\bar{K}}$.  By setting the above upper bound equal to $\delta$, a sufficient condition for meeting the spatial convergence criteria in \eqref{Eq:spa_con_cri} is 
\begin{eqnarray}\label{Eq:Spa_conv_cri_low}
\mathsf{E}[\bar{g}_0(K,N)\mid K > 0] &\leqslant& \frac{(\delta-p_{\text{null}}) \varepsilon_0}{1-p_{\text{null}}}. 
\end{eqnarray}
Consider a  \emph{typical non-empty} cell from  uniformly sampling  the set of non-empty cells. Then the spatial convergence rate of FEEL as measured using the metric $\mathsf{E}[\bar{g}_0(K,N)\mid K > 0]$ can be obtained as shown in the following theorem. 

\begin{theorem}[Spatial Convergence with Digital Transmission and Low-mobility]\label{Theo:Conv_dig_low} \emph{In this case, given the learning rate $\mu = \frac{1}{L_0\sqrt{N}}$, the expected averaged-gradient norm of a  \emph{ typical non-empty} cell is bounded as follows:
    \begin{eqnarray}
        \mathsf{E}[\bar{g}_0(K,N)\mid K > 0] \leqslant \frac{1}{\sqrt{N}} \left[ (F_0-F^*)+\frac{\sigma^2 e^{-\bar{K}}}{1-e^{-\bar{K}}} \left(\mathrm{Ei}(\bar{K})-\log \bar{K} -\gamma\right)\right],
        \label{Eq:dig_conv}
    \end{eqnarray}
where $\bar{K}$ is the expected number of active devices  in \eqref{eq:zeta}, the exponential integral $\mathrm{Ei}(x) = \int_{-\infty}^x \frac{\exp(t)}{t}dt$ and $\gamma$ represents  the Euler's Constant ($\approx 0.5772...$).}
\end{theorem}
\begin{proof}
    See Appendix \ref{App:Spa_conv_dig_low}.
\end{proof}

At the right-hand side of \eqref{Eq:dig_conv}, the first term, namely $(F_0-F^*)$, represents gradient descent along a path defined by the ground-true gradients. On the other hand, the second term that is a function of $\bar{K}$ reflects the effect of inaccurate distributed gradient estimation. Its dependance on $\bar{K}$ is discussed as follows.   According to Remark \ref{remark:fin_k}, in a network with dense devices,  $\bar{K}$ is independent of the device density  but proportional to $M \theta^{-\frac{2}{\alpha}}$. Using the result,  it follows from Theorem~\ref{Theo:Conv_dig_low} that the deviation of convergence rate from   the ideal one  can be approximated as 
\begin{equation}
    \mathsf{E}[\bar{g}_0(K,N)\mid K > 0] - \frac{F_0-F^*}{\sqrt{N}} \approx  \frac{c_1\sigma^2}{\sqrt{N}} \exp\left(-\frac{c_2 M}{\theta^{\frac{2}{\alpha}}}\right), 
\end{equation}
where $c_1$ and $c_2$ are constants. One can observe that the loss in convergence rate due to distributed gradient estimation  decays at an exponential rate when either  the number of sub-channels, $M$, or the SIR-threshold function $\theta^{-\frac{2}{\alpha}}$ increases. The gain of the former arises from interference suppression using FHSS and that of the latter from the reduction of outage probability as $\theta$ reduces, both of which contribute to the  growth of the number of successful devices. 

It should be emphasized that the above gains of  convergence rate (in round) is at the cost of increased per-round latency (in second). The learning latency is discussed as follows. 

\begin{remark}[Learning Latency] \label{remark:dig_lat_low} \emph{For ease of notation, define the  constant $\varepsilon = \frac{(\delta-p_{\text{null}}) \varepsilon_0}{1-p_{\text{null}}}$. Based on the result in Theorem \ref{Theo:Conv_dig_low}, to meet the spatial-convergence criterion in \eqref{Eq:spa_con_cri}, the expectation of the required   number of round, denoted as $N_{\min}$,  is upper bounded as 
		\begin{equation}
        \mathsf{E}[N_{\min}]\leq \frac{1}{\varepsilon^2}\left[(F_0-F^*)+\frac{\sigma^2 e^{-\bar{K}}}{1-e^{-\bar{K}}} \left(\mathrm{Ei}(\bar{K})-\log \bar{K} -\gamma\right)\right]^2. 
        \label{Eq:dig_rounds_low}
\end{equation}
Then the expected learning latency defined in \eqref{Eq:latency} is given as 
\begin{eqnarray}
        \bar{T}_{\Sigma} =         \mathsf{E}[N_{\min}]\cdot\underbrace{\left(\frac{SDM}{B\log(1+\theta)}+t_{\text{cmp}}+t_{\text{bc}}\right)}_{\text{Per-round latency}}.
        \label{Eq:dig_latency_low}
    \end{eqnarray}
where  $t_{\text{cmp}}$ and $t_{\text{bc}}$ are recalled to be   constant latency for computation and broadcasting, respectively.   The dependence of learning latency  on  network parameters are  described as follows. 
\begin{itemize}
    \item  \emph{(SIR Threshold)} Increasing the SIR threshold $\theta$ is found to have two opposite effects.  On one hand, a larger  $\theta$ reduces the number of active  devices  and increases the null probability $p_{\text{null}}$. This causes  the increase of the required  rounds  for spatial convergence.   On the other hand, increasing $\theta$ leads to a higher data rate and hence lower  per-round latency. These effects give rise to the need optimizing $\theta$ for minimizing the learning latency as further illustrated by experimental results in the sequel.   
      \item \emph{(Device Density)} One can observe from  \eqref{Eq:dig_rounds_low} and \eqref{Eq:dig_latency_low} that  the device density $\lambda_d$ (or the expected number of active devices $\bar{K}$)   affects only  the expected number of rounds but not the per-round latency. As $\lambda_d$ (or $\bar{K}$) increasing, the expected number of rounds converges to the  minimum.  
    \item \emph{(Processing Gain)} Increasing the processing gain of FHSS, $M$, reduces the number of required rounds (via increasing the number of active devices) but \emph{linearly} increases the per-round latency. When there is a sufficiently large number of active devices (i.e., sufficient exploited data), it is desirable to rein in the second effect by keeping $M$ small. 
    \end{itemize}
}
\end{remark}
\subsection{Spatial Convergence Analysis with High Mobility}
In this sub-section, we show that high mobility increases  the spatial convergence rate as well as reduces the learning latency. In this case, the typical-cell realization changes independently over rounds. Consequently, an empty cell in one round can be non-empty in another.  In contrast, the realization is fixed throughout the learning process in the case of low mobility. Therefore, for FEEL to be feasible, the typical cell in the current case should uniformly sample those cells that are \emph{non-empty in at least one of $N$ round}, i.e., $N_e > 0$. For the consistency with digital case and tractability, it is also necessary to choose a suitable learning rate as $\mu=\frac{1}{L_{0}}\sqrt{\mathsf{E}\left.\left[\frac{1}{N_{\text{e}}}\right|N_{\text{e}} \geqslant 1 \right]}$. Then the spatial convergence rate is derived as follows. 

\begin{theorem}[Spatial Convergence with Digital Transmission and High-mobility]\label{Theo:Conv_dig_high} 
 \emph{In this case, given the learning rate $\mu=\frac{1}{L_{0}}\sqrt{\mathsf{E}\left.\left[\frac{1}{N_{\text{e}}}\right|N_{\text{e}} \geqslant 1 \right]}$ and small $p_{\text{null}}$, the expected averaged-gradient norm of the typical cell that is non-empty in at least one  round is bounded as follows:
    \begin{eqnarray}
        \mathsf{E}[\bar{g}_0(N)\mid N_{\text{e}} \geqslant 1] &\leqslant \sqrt{\frac{1}{N}+  \frac{p_{\text{null}}}{N-1}}  \left[(F_0-F^*)+\frac{\sigma^2 e^{-\bar{K}}}{1-e^{-\bar{K}}} \left(\mathrm{Ei}(\bar{K})-\log \bar{K} -\gamma\right)\right] +O(p_{\text{null}}^2),
        \label{Eq:conv_dig_high}
    \end{eqnarray} 
    where  $\bar{K}$ is defined in \eqref{eq:zeta} and $\mathrm{Ei}(\cdot)$ and $\gamma$ follow those in Theorem  \ref{Theo:Conv_dig_low}.
    }
\end{theorem}
\begin{proof}
    See Appendix \ref{App:Spa_conv_dig_high}.
\end{proof}

Comparing Theorems \ref{Theo:Conv_dig_low} and \ref{Theo:Conv_dig_high}, when $p_{\text{null}}$ is small, one can conclude that high mobility slightly reduces  the spatial  convergence rate, which is  averaged over non-empty cells,  approximately by the factor of $\sqrt{1 + p_{\text{null}} }$. However, it should be emphasized that the percentage of non-empty cells in the case of high mobility is larger that that in the case of low mobility, namely $(1 - p_{\text{null}}^{N})$ versus $(1 - p_{\text{null}})$. If all cells are considered, the opposite conclusion can be  drawn based on  the following learning-latency analysis.

To this end, the above result is applied to analyzing  the learning latency in the case of high-mobility. Similar to \eqref{Eq:Spa_conv_cri_low}, a sufficient condition for meeting the spatial convergence criterion in \eqref{Eq:spa_con_cri} is obtained as 
\begin{eqnarray}\label{Eq:Spa_conv_cri_high}
\mathsf{E}[\bar{g}_0(N)\mid N_{\text{e}} \geqslant 1] &\leqslant& \frac{(\delta-p_{\text{null}}^{N}) \varepsilon_0}{1-p_{\text{null}}^{N}}. 
\end{eqnarray}
Using the condition, the learning latency is analzyed and  compared with that in the case of low  mobility as discussed in Remark \ref{remark: dig_high_low_comp}. 

\begin{remark}[Learning Latency Comparison] \label{remark: dig_high_low_comp}\emph{For convenience, define the  constant $\varepsilon' = \frac{(\delta-p_{\text{null}}^N) \varepsilon_0}{1-p_{\text{null}}^N}$. Let  $N'_{\min}$  and $\bar{T}'_{\Sigma}$ denote the required number of rounds and learning latency under the sufficient convergence conditions in  \eqref{Eq:Spa_conv_cri_high}. Then they can be derived using Theorem \ref{Theo:Conv_dig_high}. Using the result and Remark \ref{remark:dig_lat_low},  since per-round latency is identical for both the cases of low and high mobility, the ratio of corresponding expected latency is equal that of the expected numbers of required rounds: 
\begin{eqnarray}
\frac{\bar{T}'_{\Sigma}}{\bar{T}_{\Sigma}} =\frac{\mathsf{E}[N'_{\min}]}{\mathsf{E}[N_{\min}]}&\approx& \frac{\varepsilon^2\sqrt{1 +p_{\text{null}} }}{(\varepsilon')^2}\nonumber\\
&\approx& 1 - \left(\frac{2}{\delta} -\frac{5}{2}\right)p_{\text{null}}, \qquad p_{\text{null}}\rightarrow 0. \nonumber
\end{eqnarray}
As suggested by the result, if $\delta$ is small, the learning latency  (in second) with high mobility is slightly smaller  than the low-mobility counterpart despite low-mobility having a faster convergence rate (in round) in non-empty cells. The reason is that in the former case, more cells are  able to support FEEL and hence a more relaxed spatial convergence criterion. Note that the above analysis is based on approximation and bounds. Therefore, the actual quantification may not be accurate despite yielding the correct conclusion. More significant latency reduction due to high mobility is observed from experimental results in the sequel. }
\end{remark}

\section{Spatial Convergence for the Analog-Transmission Case} \label{Sec:ana}

In the preceding section, spatial convergence of the FEEL is studied for the digital-transmission case. In this section, it is analyzed for the analog-transmission case that enables low-latency over-the-air  aggregation. We assume low mobility. The extension of the results to the case of high mobility is similar to that in the preceding section. As it yields no new insight, the details are omitted for brevity. By reuse of notation, identical symbols as used in the preceding section are also used to to denote their counterparts in the current case whenever  there is no confusion.

First, the distinction of the current case is the direct exposure of the received signal, namely over-the-air aggregated gradient, to inter-cell interference. The effect can be expressed mathematically by deriving the deviation of the aggregated gradient from the ground truth as follows.  From (\ref{Eq:ana_recv_signal}), the expectation of aggregated  gradient is  an unbiased estimate of the ground truth:
\begin{eqnarray}
    \mathsf{E}\left[\bar{\bm{g}}_0^{(n)}\right] = \mathsf{E}\left[\frac{\mathbf{I}_0^{(n)} \tilde{\sigma}}{ K\sqrt{\eta}}+\frac{1}{ K} \sum\limits_{X\in K} \tilde{\bm{g}}_{X}^{(n)}\right] =  \nabla F(\bm{w}^{(n)}), 
    \label{Eq:ex}
\end{eqnarray}
and its variance can be written as:
\begin{eqnarray}
    \mathsf{E}\left[||\bm{\bar{g}}_0^{(n)}-\nabla F(\bm{w}^{(n)})||^2\right] & = & \mathsf{E}\left[\left\|\frac{\mathbf{I}_0^{(n)} \tilde{\sigma}}{ K\sqrt{\eta}} + \left(\frac{1}{ K} \sum\limits_{X\in K} \tilde{\bm{g}}_{X}^{(n)} - \nabla F(\bm{w}^{(n)})\right)\right\|^2\right] \nonumber \\
    & \leqslant & \frac{\tilde{\sigma}^2 (\mathbf{I}_0^{(n)})^2}{\eta K^2}+\frac{\sigma^2}{K},
    \label{Eq:va}
\end{eqnarray}
where $K$ is the number of active devices in the inscribed  cell, $\mathcal{D}_0$,  of the typical cell. Given  \eqref{Eq:ex} and \eqref{Eq:va}, a similar result as in Lemma \ref{lemma:Sing_cell_conv} can be obtained as follows. 

\begin{lemma}[Fixed-Cell Convergence with Analog Transmission and Low-mobility] \label {lemma:ana_conv_single}\emph{ In this case, consider a fixed  cell with a given  number of active devices, $K$,  and the learning rate $\mu = \frac{1}{L_0\sqrt{N}}$, the expected averaged-gradient norm is bounded as follows:
    \begin{eqnarray}
        \bar{g}_0(K,N) \leqslant \frac{1}{\sqrt{N}} \left( (F_0-F^*)+\frac{\sigma^2}{K}+ \frac{\tilde{\sigma}^2}{K^2 \eta N}\sum_{n = 0}^{N-1}(\mathbf{I}_0^{(n)})^2 \right).
        \label{Eq:ana_conv_single}
    \end{eqnarray}
}
\end{lemma}
Accounting for the random distribution of $K$, the  spatial-and-round averaged gradient norm follows from Lemma~\ref{lemma:ana_conv_single} as 
\begin{eqnarray}
    \!\mathsf{E}[\bar{g}_0(K,N)\mid K>0] \leqslant 
     \frac{1}{\sqrt{N}} \left( F_0-F^* + \ \sigma^2\mathsf{E} \left.\left[\frac{1}{K}\right|K>0\right] +  \frac{\tilde{\sigma}^2}{\eta}\mathsf{E}\left.\left[\frac{(\mathbf{I}_0^{(n)})^2}{K^2}\right|K>0\right]\right). \! 
    \label{Eq:ana_spa_conv_single}
\end{eqnarray}

Next, to derive a closed-form expression for the above upper bound, it is necessary to analyze the distribution of  $K$ as follows. In the digital-transmission case, a device is activated based on the criterion of successful transmission. In the current case, given truncated channel inversion in \eqref{Eq:power_control},  the criterion is for the device's fading gain to meet the truncation threshold. This results in the  \emph{activation probability} given as  $ p_{\text{a}} \triangleq \Pr(G_{X}\geqslant g_\text{th}) = e^{-g_\text{th}}$. It follows that 
\begin{eqnarray}\label{Eq: Ana_devices}
\Pr(K=j)=\frac{\exp \left(-\bar{K}'\right)\left(\bar{K}' \right)^{j}}{j !},
\end{eqnarray}
where $\bar{K}' = \pi R^2 \lambda_{\text{d}} p_{\text{a}}$ is the expected number of active  devices in the typical disk  cell, $\mathcal{D}_0$. 

Using \eqref{Eq:ana_spa_conv_single} and \eqref{Eq: Ana_devices}, we derive the main result of this section as follows. 

\begin{theorem}[Spatial Convergence with Analog Transmission and Low-mobility]\label{Theo:Conv_ana_low} \emph{
In this case, given the learning rate $\mu = \frac{1}{L_0\sqrt{N}}$, the expected averaged-gradient norm of a  typical \emph{non-empty} cell is bounded as follows:
    \begin{eqnarray}
        \!\mathsf{E}[\bar{g}_0(K,N)\mid K>0] \leqslant \frac{1}{\sqrt{N}} \left[ (F_0-F^*) + \sigma^2\phi  +  
       \frac{16 \tilde{\sigma}^2 (-\mathrm{Ei}(-g_{\text{th}}))}{p_{\text{a}}(\alpha^2-4)M}
        \left(\phi-\frac{\bar{K}'e^{-\bar{K}'}}{1-e^{-\bar{K}'}}\right) \right], \!
        \label{Eq:ana_conv}
    \end{eqnarray} 
    with  
    \begin{eqnarray}
\phi = \mathsf{E}\left.\left[\frac{1}{K}\right|K>0\right] = \frac{e^{-\bar{K}'}}{1-e^{-\bar{K}'}} \left[\mathrm{Ei}(\bar{K}')-\log(\bar{K}')-\gamma\right],
        \label{Eq:one_over_k_ana}
    \end{eqnarray}
where the expected number of active devices  $\bar{K}' = \pi R^2 \lambda_{\text{d}} p_{\text{a}}$,  the exponential integral $\mathrm{Ei}$ and Euler's Constant $\gamma$ follow those in Theorem~\ref{Theo:Conv_dig_low}, and the term 
$(-\mathrm{Ei}(-g_{\text{th}}))$ is positive. } 
\end{theorem}
\begin{proof}
See Appendix \ref{App:pert}.
\end{proof}

The second term on the right-hand side of \eqref{Eq:ana_conv}, $\sigma^2\phi$, represents the error of distributed gradient estimation and is observed to have the same form as its counterpart for the digital-transmission case in Theorem~\ref{Theo:Conv_dig_low} but with $\bar{K}$ replaced by $\bar{K}'$. Due to the different scalings of $\bar{K}$ and  $\bar{K}'$ w.r.t. the device density $\lambda_d$, there is an important difference between the two cases.  Specifically, as the density $\lambda_d$ increases, the term for the case of analog transmission   diminishes at an \emph{exponential rate}  while its  digital-transmission counterpart  converges to a \emph{constant} according to \eqref{Eq:mean_device_den}. This results  in  different accuracies of  distributed gradient estimation. On the other hand, analog transmission exposes learning to the effect of inter-cell interference as represented by the last term in \eqref{Eq:ana_conv}. Though higher density will cause larger interference, one can observe that this term also decays at an exponential rate as $\lambda_d$ grows. The fundamental reason is that more devices are involved with the increasing density and the interference can be effectively suppressed by gradient aggregation. Combining the above discussion suggests that analog transmission is preferred to digital transmission in a network with dense devices as also corroborated by experimental results. 

Next, we compare the relative effects of interference and distributed-data induced  gradient deviations from the ground truth. To this end, we consider  the following ratio between the last two terms of \eqref{Eq:ana_conv}, called \emph{interference effect}:
\begin{equation}
    \frac{\text{Interference induced deviation}}{\text{Data induced deviation}} = \frac{16 \tilde{\sigma}^2(-\mathrm{Ei}(-g_{\text{th}}))}{\sigma^2 p_{\text{a}}(\alpha^2-4)M}\cdot
    \left(1-\frac{\bar{K}'}{\mathrm{Ei}(\bar{K}')-\log(\bar{K}')-\gamma}\right).
\label{Eq:ici_effect}
\end{equation}
 The dependence of interference effect on different network parameters is discussed  as follows. 
\begin{itemize}

\item \emph{(Device-density/cell-size)} Increasing the device density or cell size both lead to linear growth of the expected number of devices, $\bar{K}'$. This reduces the interference effect in two aspects. One is the suppression of interference by more aggressively averaging via over-the-air aggregation. The other is larger path-loss for interference signals received at the BS. Mathematically, the interference reduction by increasing the cell size is reflected in the last term on the right hand side of \eqref{Eq:ici_effect}, $ \left(1-\frac{\bar{K}'}{\mathrm{Ei}(\bar{K}')-\log(\bar{K}')-\gamma}\right)$, being  a decreasing function of $\bar{K}'$. 

\item \emph{(Path-loss exponent)} The interference effect is observed to diminish as the path-loss exponent $\alpha$ increases, which reduces inter-cell interference by reducing spatial coupling between cells. 
\item \emph{(Processing gain)} The interference effect is inversely proportional to the processing gain of FHSS, $M$. Though the increase of $M$ seems to accelerate learning (in terms of rounds), it increases  per-round latency (in second) as the effective transmission bandwidth, namely $B/M$, reduces. See more discussion  in the sequel. 

\item \emph{(Channel truncation threshold)}  The interference effect decreases as a decreasing   threshold causes   the activation probability to grow. This holds only in the considered interference-limited regime. Reducing the  threshold  may not be desired in the noise-limited regime as it can cause devices with weak channels to participate in learning, amplifying the noise effect.  
\end{itemize}

Similar to the analysis shown in the digital case, given the spatial convergence target $\delta$ and the null probability $p_\text{null}$, a sufficient condition for meeting the spatial convergence criterion in \eqref{Eq:spa_con_cri} is obtained as 
\begin{eqnarray}\label{Eq:Spa_conv_cri_ana_low}
\mathsf{E}[\bar{g}_0(K,N)\mid K > 0] &\leqslant& \frac{(\delta-p_{\text{null}}) \varepsilon_0}{1-p_{\text{null}}}. 
\end{eqnarray}

\begin{remark}[Learning Latency] \emph{Under the sufficient condition  in \eqref{Eq:Spa_conv_cri_ana_low}, the expected minimum number of rounds, denoted as $\mathsf{E}[N_{\min}]$, has no simple form but can be upper bounded by the ratio between the upper bound on the averaged-gradient norm in Theorem \ref{Theo:Conv_ana_low} and  the  constant $\varepsilon = \frac{(\delta-p_{\text{null}}) \varepsilon_0}{1-p_{\text{null}}}$. The expected learning latency can be written  as 
\begin{eqnarray}
    \bar{T}_{\Sigma} = \mathsf{E}[N_{\min}]\cdot\underbrace{\left(\frac{SM}{B}+t_{\text{cmp}}+t_{\text{bc}}\right)}_{\text{Per-round latency}},
\end{eqnarray}
where $t_{\text{cmp}}$ and $t_{\text{bc}}$ are recalled to be constant computation and broadcasting latency, respectively.  One key observation is that increasing the processing gain $M$ increases per-round latency but reduces the expected number of rounds as mentioned earlier. This suggests the need of optimizing $M$ for latency minimization. 
}
\end{remark}

\section{Experimental Results}\label{Sec:Simu}

\subsection{Experimental Settings}
The experimental  settings are as follows unless specified otherwise. Consider a cellular network  in a $50 \times 50$ (unit area) horizontal area.  Each hexagon  cell's radius is  $1$ (unit length). FEEL is deployed in the cell located at the centre of the area. The path-loss exponent is set as $ \alpha = 4$, and total bandwidth is $B = 1 $ MHz. In the digital-transmission case, we assume that each coefficient  of a transmitted gradient is quantized into $16$ bits;  in the analog-transmission case, each coefficient  is mapped to a symbol. Transmission power in digital case is given by $P = 1$ for all the edge devices, while in the analog case, $\eta$ and $g_{\text{th}}$ are set to satisfy the average power constraint  $\mathsf{E}[P_X] = P$. The constant computing-and-broadcasting latency is assumed negligible in our experiments. Let each sample path be a sequence of typical-cell realizations over rounds. Then each result on spatially  averaged learning performance (i.e.,  test accuracy or learning latency)  is computed as the average of  $10$ sample paths to account for spatial  network distribution. 

The learning task is to perform the handwritten-digit recognition using the well-known MNIST dataset. There are total $60,000$ labeled training data samples in this dataset, each edge device is assigned $200$ samples by randomly sampling the  dataset. The classifier model is implemented using a $6$-layer \textit{convolutional neural network} (CNN) that consists of two $5 \times 5$ convolution layers with ReLu activation, each followed by $2 \times 2$ max pooling, a fully-connected layer with $512$ units, ReLu activation, and a final softmax output layer. 

\subsection{Effect of Device Mobility}

Consider the case of digital transmission. The curves of spatially averaged test accuracy versus the number of rounds are plotted in Fig. \ref{Fig:dev_mob} for both the cases of low and high mobility. Overall, one can observe  that  convergence rate with high mobility is faster than the low-mobility counterpart, which is aligned with the theoretic analysis.  In particular, when devices are sparse (i.e., $\lambda_d = 1$), the test accuracy with low mobility ($0.9$) is substantially lower than that with high mobility  ($> 0.95$). The reason is that the data size  and diversity are both insufficient, which, however, can be effectively overcome by mobility. The benefit of mobility in terms of convergence rate can also be observed even for a higher density, i.e., $\lambda_d = 5$. The difference attributed to  mobility  diminishes when the density is sufficiently high ($\lambda_d = 10$). 

\begin{figure*}[t!]
    \centering
    \includegraphics[scale=.5]{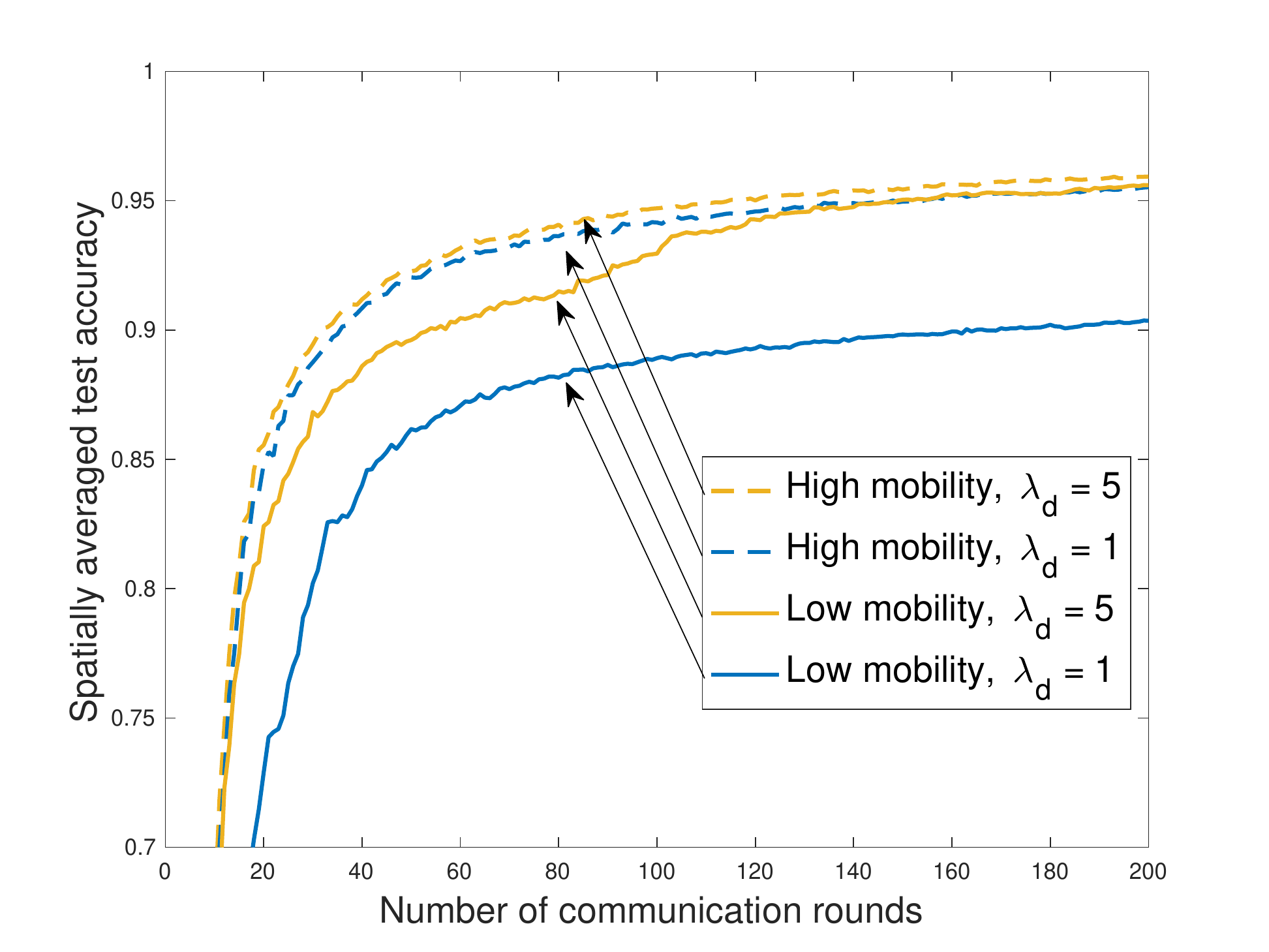}
    \caption{The spatially averaged test  accuracy versus the number of communication rounds for both the cases of  digital transmission  with low and high mobility and a varying device density $\lambda_d$.}
    \label{Fig:dev_mob}
\end{figure*}

\subsection{Effects of  Network Parameters}
Consider the case of digital transmission. The effects of network parameters, namely the device density, SIR threshold, and processing gain, on learning latency  are demonstrated in Fig. \ref{Fig:net_para}. The curves of spatially averaged learning latency versus network parameters are plotted for achieving the  target spatially averaged test accuracy of $95 \%$. Several  observations can be made. First, one can observe from  Fig.~\ref{Fig:net_para_a} that the  learning latency decreases and then saturates as $\lambda_d$ increases. The first part corresponds to the data-limited regime where increasing the density of devices contributes more training data and thereby reducing the needed number of rounds. The second part corresponds to the data-sufficient regime where more devices no longer yield an increase of the convergence rate. Second, it can be observed from Fig. \ref{Fig:net_para_b}  that the   latency first decreases  and then increases as the SIR threshold $\theta$ grows. This corroborates  Remark \ref{remark:dig_lat_low} based on analysis and suggests the need of optimizing $\theta$. Last, Fig.~\ref{Fig:net_para_c} shows the linear growth of  latency as the processing gain $M$ increases. Thus, for the current experimental settings, the minimum processing gain ($M=1$) is desired. This is aligned with Remark \ref{remark:dig_lat_low}. 

\begin{figure}[t!]
    \centering
    \subfigure[Effect of device density $\lambda_d$]
    {
        \includegraphics[width=.47\linewidth]{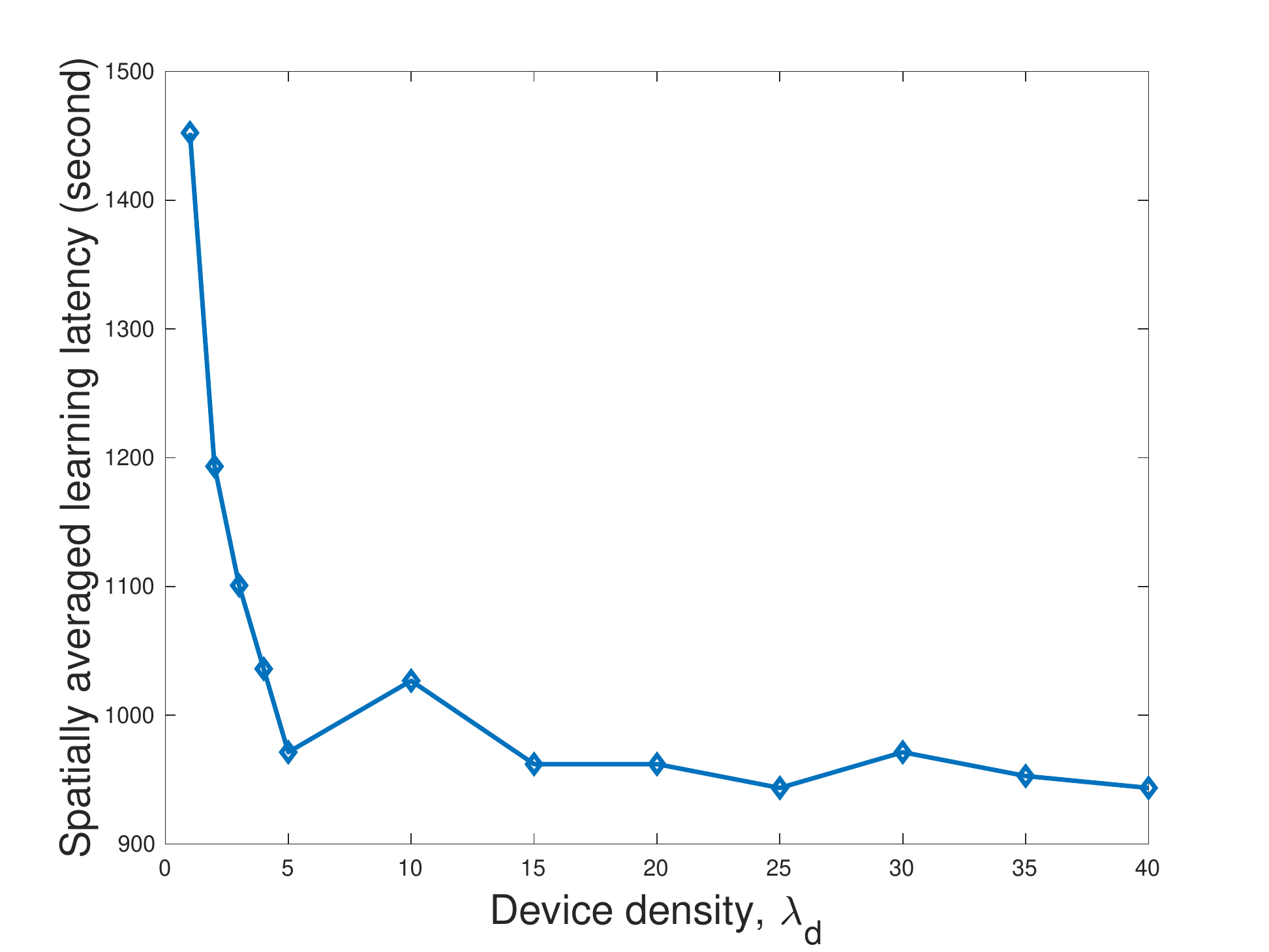}
        \label{Fig:net_para_a}
    }
    \subfigure[Effect of SIR threshold $\theta$]
    {
        \includegraphics[width=.47\linewidth]{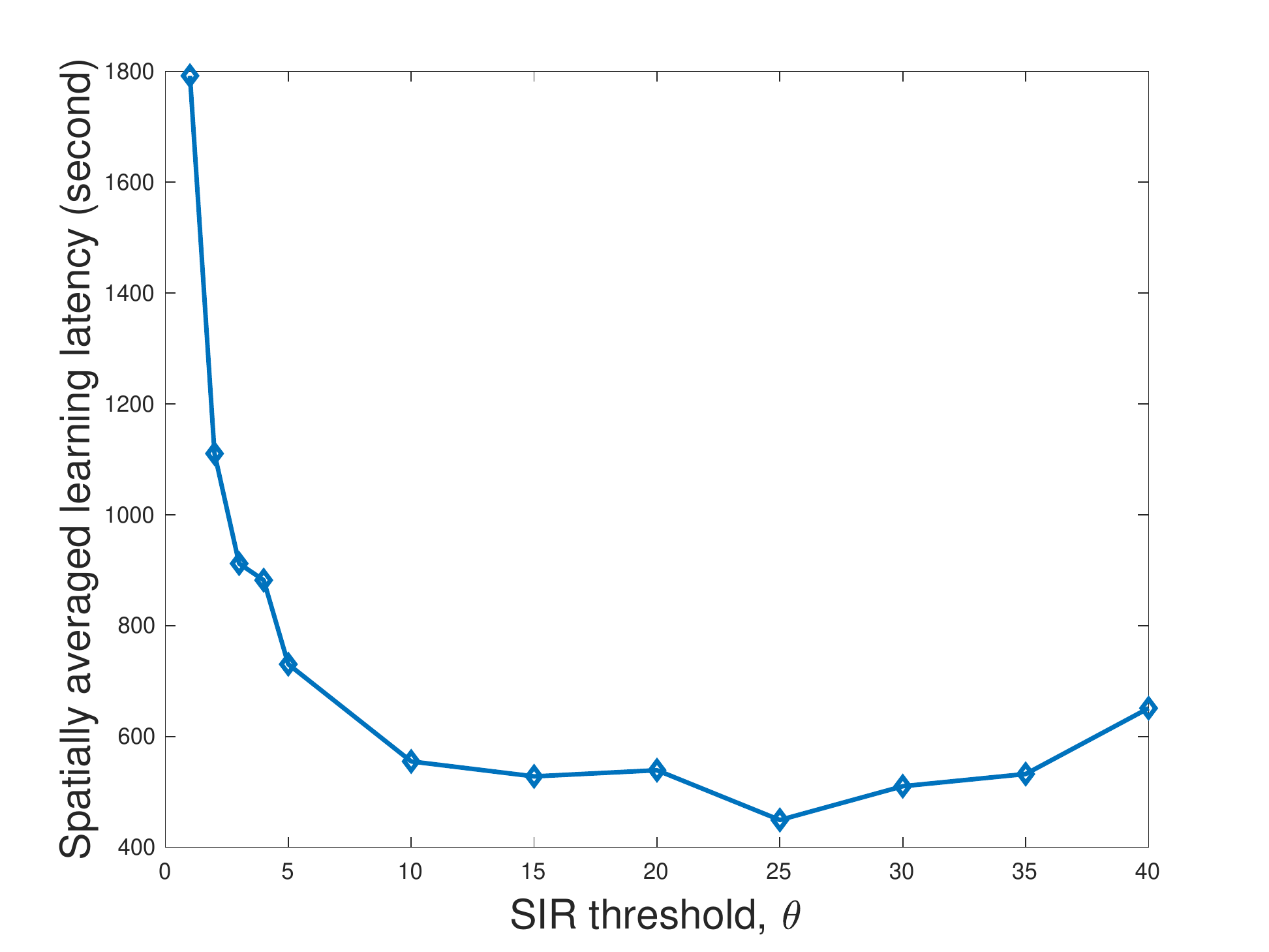}
        \label{Fig:net_para_b}
    }\\
    \subfigure[Effect of processing gain $M$]
    {
        \includegraphics[width=.47 \linewidth]{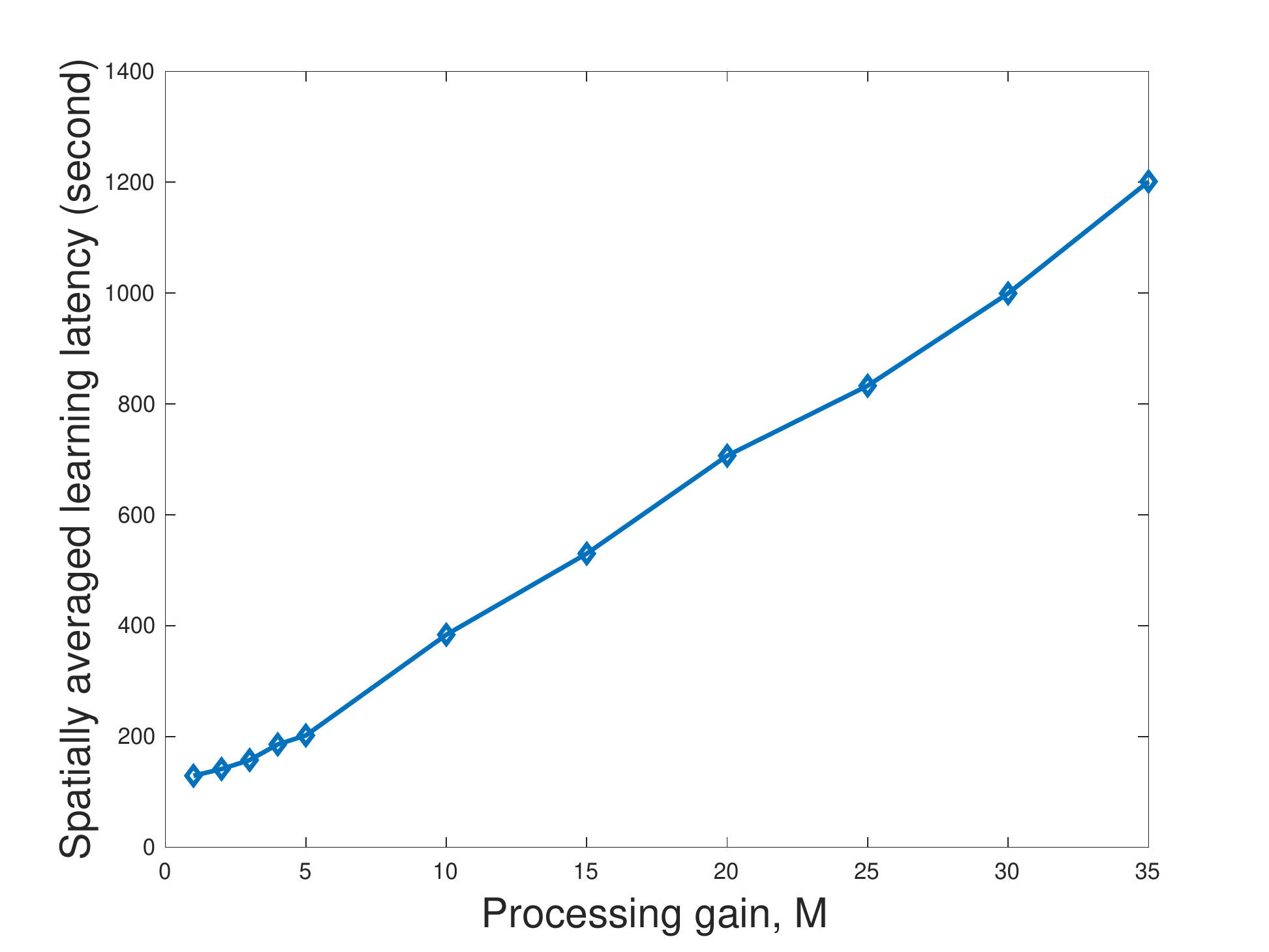}
        \label{Fig:net_para_c}
    }
    \caption{The effects of network parameters on learning latency with digital transmission for achieving  a target spatially averaged test accuracy of $95 \%$.}
    \label{Fig:net_para}
\end{figure}

\subsection{Comparison of Digital and Analog Transmission}
The spatially averaged test accuracies  for the cases of digital and analog transmission are compared in Fig.~\ref{Fig: dia_ana_comp_acc} in terms of spatially averaged test accuracy. Different device densities are considered. When the network is relatively sparse (i.e., $\lambda_d = 1$ or $3$), digital transmission is observed to outperform the analog scheme as the latter exposes uncoded signals to the perturbation of inter-cell interference. On the other hand, when there are many active devices (i.e., $\lambda_d = 30$), the aggressive over-the-air aggregation realized by analog transmission effectively suppresses interference by averaging. Consequently,  analog transmission  achieves better performance in this case. 

\begin{figure*}[t!]
    \centering
    \includegraphics[scale=.25]{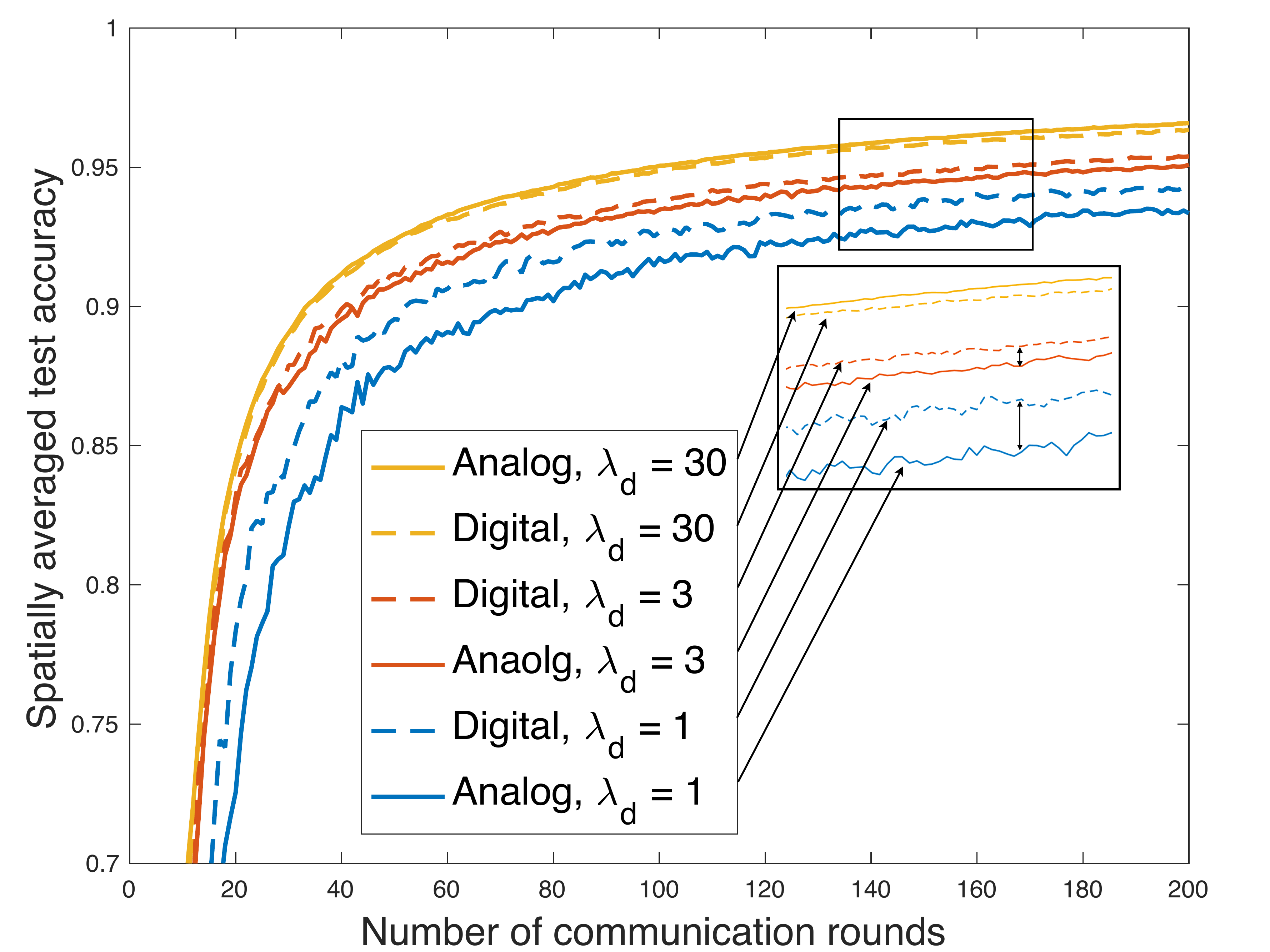}
    \caption{Comparison of spatially averaged test accuracies for the cases of digital and analog transmission for different device densities. }
    \label{Fig: dia_ana_comp_acc}
\end{figure*}

The learning performance for the cases of digital and analog transmission is further compared in Fig.~\ref{Fig: dia_ana_comp_lat} in terms of the required number of rounds and spatially averaged learning latency for the target spatially averaged test accuracy of $93 \%$. The learning latency for analog transmission is observed from  Fig.~\ref{Fig: dia_ana_comp_lat_a} to be much lower than the digital-transmission counterpart for both low-and-high device densities. One can observe from  Fig.~\ref{Fig: dia_ana_comp_lat_b} that in a sparse network ($\lambda_d \leq 10$), error-free transmission of digital transmission reduces the required number of rounds; in a dense network ($\lambda_d \geq 10$), the gain varnishes as analog transmission supports more active devices.  Regardless of this difference in terms of required rounds, the advantage of shorter per-round latency of analog transmission dominates, resulting in  the earlier observation from Fig.~\ref{Fig: dia_ana_comp_lat_a}.

\begin{figure}[t!]
    \centering
    \subfigure[Spatially Averaged Learning Latency]
    {
        \includegraphics[width=.47\linewidth]{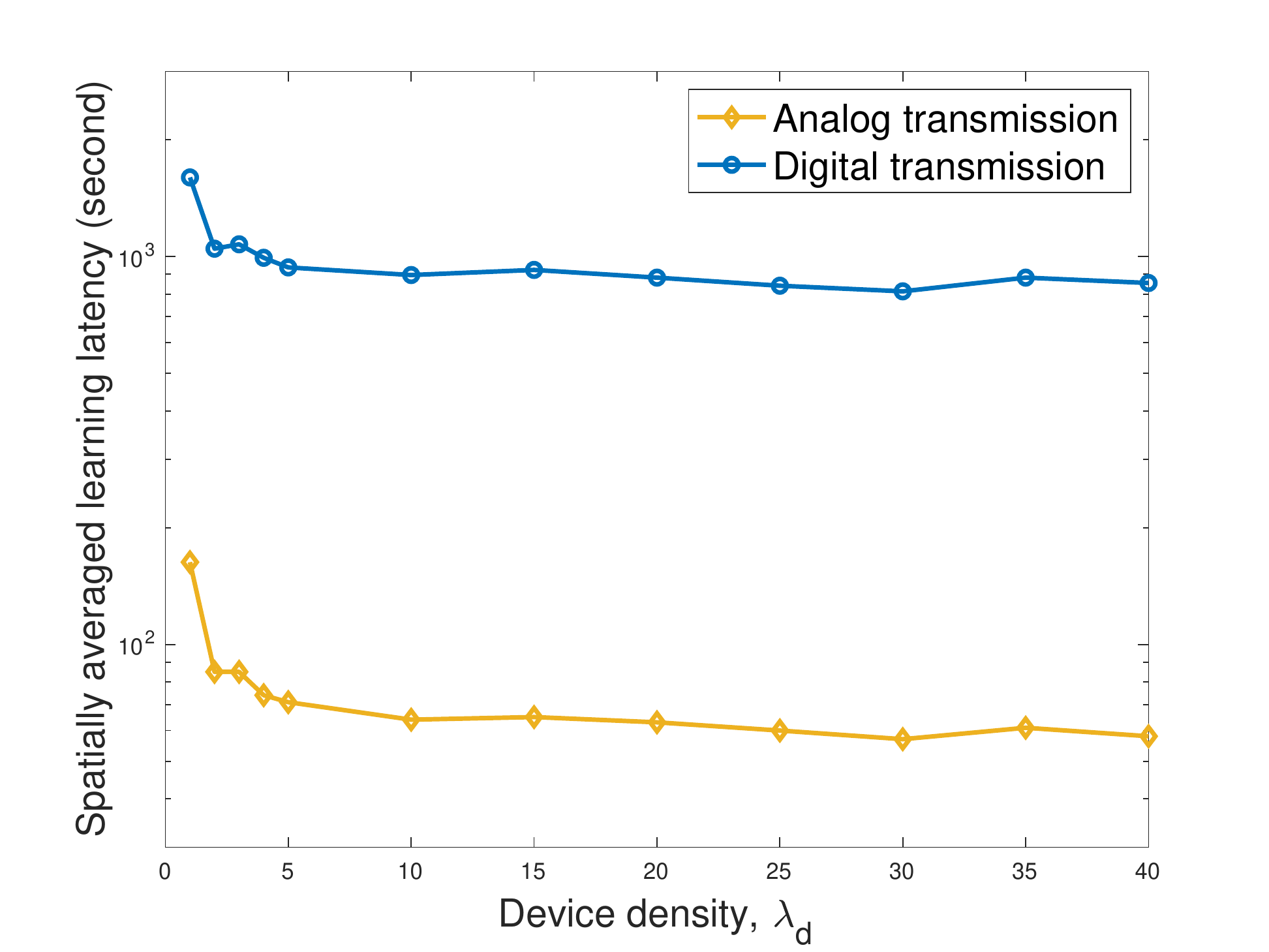}
        \label{Fig: dia_ana_comp_lat_a}
    }
        \subfigure[Required Number of Communication Rounds]
    {
        \includegraphics[width=.47\linewidth]{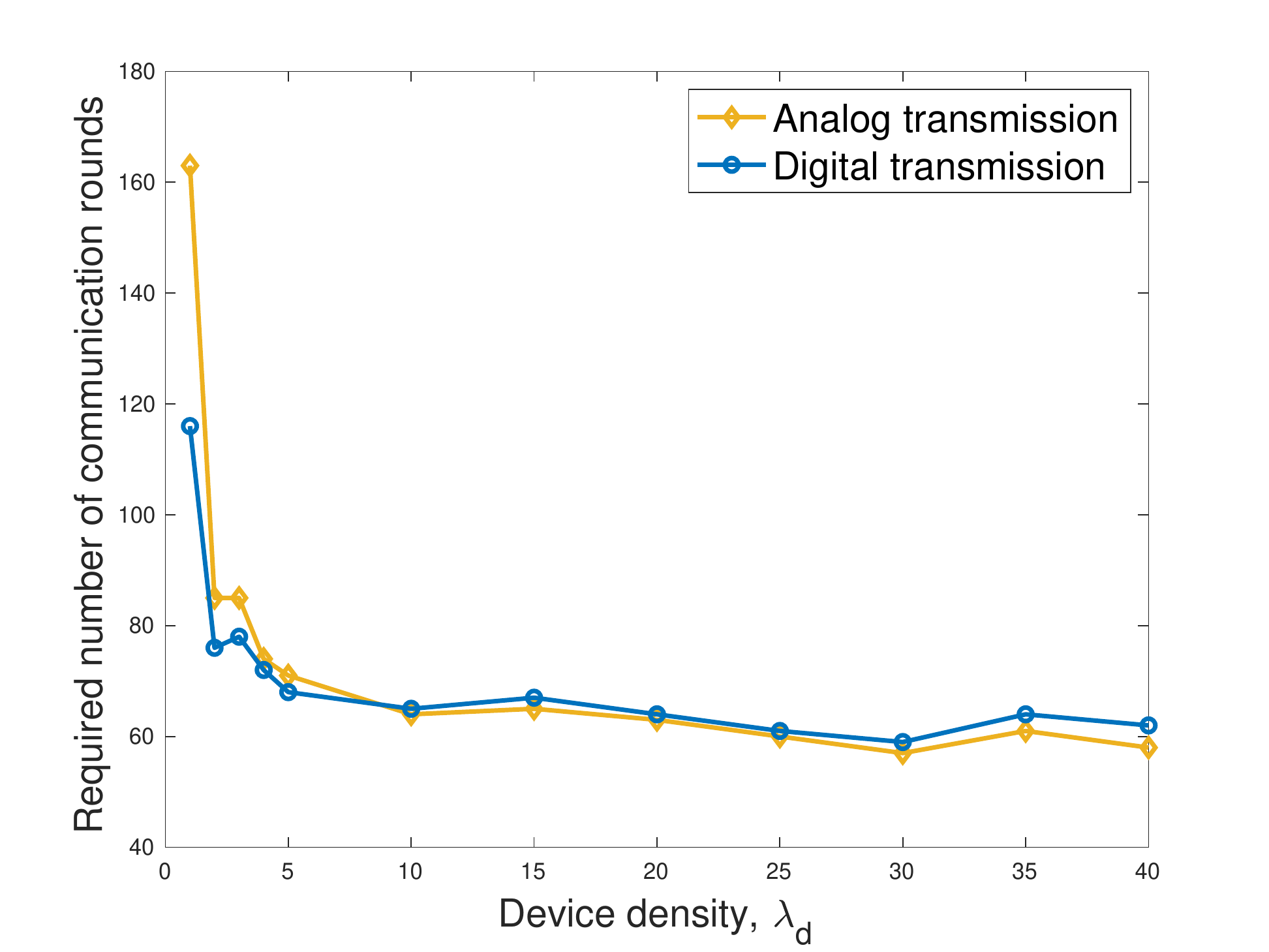}
        \label{Fig: dia_ana_comp_lat_b}
    }

    \caption{Learning-performance comparison between  the cases of digital and analog transmission in terms of: (a) spatially averaged learning latency  and (b) the required number of rounds for the targeted spatially averaged test accuracy of $95\%$.}
    \label{Fig: dia_ana_comp_lat}
\end{figure}

\section{Concluding Remarks}

In this work, we have investigated the spatial convergence  of  FEEL deployed in a typical cell of a large-scale cellular  network. Both the schemes of digital and analog transmission are considered. In terms of spatial convergence rate (in round), digital transmission is preferred for low-to-medium device densities while convergence with analog transmission is faster when devices are dense. On the other hand, in terms of learning latency (in second), analog transmission is always preferred due to its support of low-latency over-the-air aggregation. 

This work opens the direction of distributed edge learning in large-scale cellular networks, in which numerous topics warrant further investigation. In particular, a more complex network topology including both backhaul and radio-access links  can be considered to support hierarchical federated learning involving edge devices, edge servers and central-could servers. Moreover, the current work suggests the need of optimizing network parameters (e.g., SIR threshold), which can be further investigated to improve the learning performance. Furthermore, for the deployment of FEEL in 5G networks,  it is interesting to study the effects of advanced physical-layer techniques (i.e., massive MIMO and non-orthogonal access) on the spatial learning performance.

\appendices
\section{Proof of theorem \ref{Theo:Conv_dig_low}}\label{App:Spa_conv_dig_low}
It follows from Lemma \ref{lemma:Sing_cell_conv} that 
\begin{equation}
\mathsf{E}[\bar{g}_0(K,N)\mid K > 0] \leqslant \frac{1}{\sqrt{N}} \left( (F_0-F^*)+\sigma^2\mathsf{E}\left.\left[\frac{1}{K}\right|K>0\right]\right). 
\label{a1}
\end{equation}
Using the distribution of $K$ in Lemma \ref{lemma:dis_dig_device}, the last term in \eqref{a1} can be obtained as follows
    \begin{eqnarray}
        \mathsf{E}\left.\left[\frac{1}{K}\right|K>0\right] &=& \sum\limits_{j = 1}^\infty \frac{1}{j}p(A = j|A>0) \nonumber \\
        & = & \frac{\exp(-\pi R^2\lambda_{\text{d}} p_{\text{s}})}{1-\exp(-\pi R^2\lambda_{\text{d}} p_{\text{s}})}\sum\limits_{j = 1}^\infty \frac{1}{j} \frac{(\pi R^2\lambda_{\text{d}} p_{\text{s}})^j}{j!} \nonumber \\ 
        & \stackrel{(b)}{=} & \frac{\exp(-\pi R^2\lambda_{\text{d}} p_{\text{s}})}{1-\exp(-\pi R^2\lambda_{\text{d}} p_{\text{s}})} \left(\mathrm{Ei}(\pi R^2\lambda_{\text{d}} p_{\text{s}})-\log(\pi R^2\lambda_{\text{d}} p_{\text{s}})-\gamma\right),
        \label{Eq:E_1_K}
    \end{eqnarray}
where (b) is obtained using \cite[(3.16)]{chao1972negative}. The desired result follows.

\section{Proof of Theorem \ref{Theo:Conv_dig_high}}\label{App:Spa_conv_dig_high}
First, we analyze  the convergence in the  typical cell condition on its being non-empty. Based on Assumption 2 and substituting \eqref{Eq:global_update} into \eqref{Eq:ass2},   the single-step improvement of the loss function is obtained as
\begin{eqnarray}
    F(\bm{w}^{(n+1)})-F(\bm{w}^{(n)}) &\leqslant& \left(\nabla F(\bm{w}^{(n)})\right)^{T} \left(\bm{w}^{(n+1)}-\bm{w}^{(n)}\right)+\sum_{i=1}^{d} \frac{L_{i}}{2}\left(\bm{w}^{(n+1)}-\bm{w}^{(n)}\right)_i^{2} \nonumber\\
    &\leqslant& -\mu \left(\nabla F(\bm{w}^{(n)})\right)^{T} \bar{\bm{g}}^{(n)} + \frac{\|L\|_{\infty}}{2} \mu^2 \|\bar{\bm{g}}^{(n)}\|^2,
    \label{c1}
\end{eqnarray}
where $\bar{\bm{g}}^{(n)}$ is the aggregated gradient received at the BS in the  $n^{th}$  round. Note that two random processes underpinning the spatial learning process are  $\{K^{(n)}; n\geqslant 0\}$ and $\{\bar{\bm{g}}^{(n)}; n\geqslant 0\}$. Consider the $(n+1)^{th}$ communication round, condition on fixed $K^{(n)} > 0$ and the model updated  in the preceding round, taking expectation of  both sides of (\ref{c1}) yields 
\begin{eqnarray}
    \mathsf{E}\left[F(\bm{w}^{(n+1)})-F(\bm{w}^{(n)})| K^{(n)} \right] &\leqslant& -\mu \left(\nabla F(\bm{w}^{(n)})\right)^{T} \left( \frac{1}{K^{(n)}} \sum\limits_{X\in \mathcal{D}_0^{(n)}}  \mathsf{E}\left[\tilde{\bm{g}}_{X}^{(n)}\right]\right) \nonumber\\
    &+& \frac{L_0\mu^2}{2} \mathsf{E} \left. \left[\|\frac{1}{K^{(n)}} \sum\limits_{X\in \mathcal{D}_0^{(n)}} \tilde{\bm{g}}_{X}^{(n)}\|^2 \right|K^{(n)}\right],
    \label{c2}
\end{eqnarray}
where $L_0 = \|L\|_{\infty}$. Based on  Assumption 3 and Lemma 3, (\ref{c2}) can be written as:
\begin{eqnarray}
    \mathsf{E}\left[F(\bm{w}^{(n+1)})-F(\bm{w}^{(n)})| K^{(n)} \right] &\leqslant&
    -\mu \|\nabla F(\bm{w}^{(n)})\|^2 + \frac{L_0\mu^2}{2} \left( \|\nabla F(\bm{w}^{(n)})\|^2 + \frac{\sigma^2}{K^{(n)}}\right)  \nonumber \\
    & = &
    \left(-\mu + \frac{L_0\mu^2}{2}\right)\|\nabla F(\bm{w}^{(n)})\|^2 +\frac{L_0\mu^2\sigma^2}{2K^{(n)}}.
\end{eqnarray}
Conditioning on the effective number of rounds $N_{\text{e}}\geq 1$, performing  a telescoping sum over the iterations gives
\begin{eqnarray}
F_0-F^* & \geq & F_0 - \mathsf{E} [F(\bm{w}^{(n)})|N_{\text{e}}\geq 1]  \nonumber \\
& = & \mathsf{E}\left.\left[\sum_{n = 0}^{N_{\text{e}}-1} \left(F(\bm{w}^{(n)})-F^{(n+1)}\right)\right|N_{\text{e}}\geq 1\right] \nonumber \\
& \geq &  \left(\mu - \frac{L_0\mu^2}{2}\right)\sum_{n = 0}^{N_{\text{e}}-1}\|\nabla F(\bm{w}^{(n)})\|^2 - \frac{L_0\mu^2\sigma^2}{2}\sum_{n = 0}^{N_{\text{e}}-1}\frac{1}{K^{(n)}}, \qquad N_{\text{e}}\geq 1.
\label{c4}
\end{eqnarray}
Since  $N_{\text{e}}$ is a random variable, it follows that 
\begin{eqnarray}
    &~& \!(F_0-F^*)\mathsf{E}\left.\left[\frac{1}{N_{\text{e}}} \right| N_{\text{e}}\geqslant 1\right]\geqslant \nonumber \\ 
    &  & \left(\mu - \frac{L_0\mu^2}{2}\right) \mathsf{E}\left.\left[\frac{1}{N_{\text{e}}}\sum_{n = 0}^{N_{\text{e}}-1}\|\nabla F(\bm{w}^{(n)})\|^2\right|N_{\text{e}}\geqslant 1\right]-\frac{L_0\mu^2\sigma^2}{2}\mathsf{E}\left.\left[\frac{1}{N_{\text{e}}}\sum_{n = 0}^{N_{\text{e}}-1}\frac{1}{K^{(n)}}\right|N_{\text{e}}\geqslant 1\right]. \nonumber\!
\end{eqnarray}
By rearranging  the terms, 
\begin{eqnarray}
    &&\mathsf{E}\left.\left[\frac{1}{N_{\text{e}}}\sum_{n = 0}^{N_{\text{e}}-1}\|\nabla F(\bm{w}^{(n)})\|^2\right|N_{\text{e}} \geqslant 1\right] \leqslant \nonumber \\
    && \qquad\qquad \frac{(F_0-F^*)\mathsf{E}\left.\left[\frac{1}{N_{\text{e}}} \right| N_{\text{e}}\geqslant 1\right] + \frac{L_0\mu^2\sigma^2}{2}\mathsf{E}\left.\left[\frac{1}{N_{\text{e}}}\sum_{n = 0}^{N_{\text{e}}-1}\frac{1}{K^{(n)}}\right|N_{\text{e}}\geqslant 1\right]}{\mu-L_0\mu^2/2}.     \label{c5a}
\end{eqnarray}
On the other hand, since $\mu = \frac{1}{L_0}\left(\mathsf{E}\left.\left[\frac{1}{N_{\text{e}}} \right| N_{\text{e}}\geqslant 1\right]\right)^{1/2}$, 
\begin{eqnarray}
    \frac{1}{\mu-L_0\mu^2/2} &=& \frac{2L_0}{\left(\mathsf{E}\left.\left[\frac{1}{N_{\text{e}}} \right| N_{\text{e}}\geqslant 1\right]\right)^{1/2}\left(2-\left(\mathsf{E}\left.\left[\frac{1}{N_{\text{e}}} \right| N_{\text{e}}\geqslant 1\right]\right)^{1/2}\right)} \nonumber \\
    &\overset{(a)}{\leqslant}&\frac{2L_0}{\left(\mathsf{E}\left.\left[\frac{1}{N_{\text{e}}} \right| N_{\text{e}}\geqslant 1\right]\right)^{1/2}},
    \label{c5}
\end{eqnarray}
where (a) follows from $\left(\mathsf{E}\left.\left[\frac{1}{N_{\text{e}}} \right| N_{\text{e}}\geqslant 1\right]\right)^{1/2}\leqslant 1$. By combining \eqref{c5a} and \eqref{c5}, and replacing \!$\mathsf{E}\left.\left[\frac{1}{N_{\text{e}}}\sum_{n = 0}^{N_{\text{e}}-1}\|\nabla F(\bm{w}^{(n)})\|^2\right|N_{\text{e}}\geqslant 1\right]$ with $\bar{g}_0(N)$, 
\begin{eqnarray}
    \bar{g}_0(N)  \leqslant \left(\mathsf{E}\left.\left[\frac{1}{N_{\text{e}}} \right| N_{\text{e}}\geqslant 1\right]\right)^{1/2} \left( (F_0-F^*)+\sigma^2\mathsf{E}\left.\left[\frac{1}{N_{\text{e}}}\sum_{n = 0}^{N_{\text{e}}-1}\frac{1}{K^{(n)}}\right| N_{\text{e}}\geqslant 1\right] \right).
\end{eqnarray}
For spatial convergence, take expectation over  the spatial  distribution of edge devices, one can obtain the following upper bound on the spatial-and-round averaged gradient: 
\begin{eqnarray}
    \mathsf{E}[\bar{g}_0(N)\mid N_{\text{e}} \geqslant 1] \leqslant \left(\mathsf{E}\left.\left[\frac{1}{N_{\text{e}}} \right| N_{\text{e}}\geqslant 1\right]\right)^{1/2} \left( (F_0-F^*)+\sigma^2\mathsf{E}\left.\left[\frac{1}{K^{(n)}}\right|K^{(n)}>0\right]\right). \label{Eq:ExpGrad:a}
\end{eqnarray}
The expression for the term in \eqref{Eq:ExpGrad:a},  $ \mathsf{E}\left.\left[\frac{1}{N_{\text{e}}} \right| N_{\text{e}}\geqslant 1\right]$,  can be obtained as
\begin{eqnarray}
    \mathsf{E}\left.\left[\frac{1}{N_{\text{e}}} \right| N_{\text{e}}\geqslant 1\right]
    &\!=\!& \sum_{i=1}^{N} \frac{1}{i} \left(\begin{array}{c}
        N \\
        i
        \end{array}\right)
    \frac{\left(1-p_{\text{null}}\right)^{i} \left( p_{\text{null}}\right)^{N-i}}{1 - p_{\text{null}}^{N}} \nonumber\\
    &\overset{\text{(a)}}{=}& \frac{1}{1 - p_{\text{null}}^{N}} \sum\limits_{i = 1}^{N}\frac{p_{\text{null}}^{i-1}-p_{\text{null}}^N}{N-i+1}  \nonumber \\
    &=& \frac{1}{N}+ \frac{p_{\text{null}}}{N-1}+O(p_{\text{null}}^2), \qquad p_{\text{null}}\rightarrow 0, \label{Eq:ExpGrad:b}
\end{eqnarray}
where (a) is based on (10) in \cite{stephan1945expected}. Substituting \eqref{Eq:ExpGrad:b} and \eqref{Eq:E_1_K} into  \eqref{Eq:ExpGrad:a} yields the desired result. 

\section{Proof of Theorem \ref{Theo:Conv_ana_low}}\label{App:pert}
Starting from \eqref{Eq:ana_spa_conv_single}, the proof focuses on deriving an expression for the  perturbation term caused by inter-cell interference, namely $\frac{\tilde{\sigma}^2}{\eta}\mathsf{E}\left.\left[\frac{(\mathbf{I}_0^{(n)})^2}{K^2}\right|K>0\right]$. Due to the independence between the interference $(\mathbf{I}_0^{(n)})^2$ and the number of  devices $K$, 
\begin{equation}
\frac{\tilde{\sigma}^2}{\eta}\mathsf{E}\left.\left[\frac{(\mathbf{I}_0^{(n)})^2}{K^2}\right|K>0\right]= \frac{\tilde{\sigma}^2}{\eta}\mathsf{E}\left[(\mathbf{I}_0^{(n)})^2\right]\mathsf{E}\left.\left[\frac{1}{K^2}\right|K>0\right]. \label{Eq:Interf}
\end{equation}
First, by applying Campbell's Theorem \cite{kingman2005p}, 
\begin{eqnarray}
    \mathsf{E}\left[(\mathbf{I}_0^{(n)})^2\right] = \frac{2\pi\lambda_{\text{d}} P R^{2-\alpha}}{(\alpha-2)M}.
    \label{Eq:epi}
\end{eqnarray}
Next, 
\begin{eqnarray}
    \mathsf{E}\left.\left[\frac{1}{K^2}\right|K>0\right] & = & \mathsf{E}\left.\left[\left(1+\frac{1}{K}\right)^2\cdot\frac{1}{(K+1)^2}\right|K>0\right] \nonumber \\
    &\leqslant& 4\mathsf{E}\left.\left[\frac{1}{(K+1)^2}\right|K>0\right].
\end{eqnarray}
Using the distribution of $K$ in \eqref{Eq: Ana_devices}, 
\begin{eqnarray}
\mathsf{E}\left.\left[\frac{1}{K^2}\right|K>0\right] 
    & \leq & \frac{4}{1-\exp(-\bar{K}')}\sum\limits_{j = 1}^\infty \frac{\exp(-\bar{K}')}{(j+1)^2} \frac{(\bar{K}')^j}{j!} \nonumber \\
    & = & \frac{4}{1-\exp(-\bar{K}')}\left(\sum\limits_{j = 0}^\infty \frac{\exp(-\bar{K}')}{(j+1)^2} \frac{(\bar{K}')^j}{j!}-\exp(-\bar{K}')\right)
    \nonumber \\
    & = & \frac{4}{1-\exp(-\bar{K}')}\left(\frac{1}{\bar{K}'}\sum\limits_{j = 1}^\infty \frac{\exp(-\bar{K}')}{j} \frac{(\bar{K}')^j}{j!}-\exp(-\bar{K}')\right) \nonumber \\
    & = & \frac{4}{\bar{K}'}\mathsf{E}\left.\left[\frac{1}{K}\right|K>0\right]-\frac{\exp(-\bar{K}')}{1-\exp(-\bar{K}')}.
    \label{Eq:one_over_k_ana_sq}
\end{eqnarray}
Combining \eqref{Eq:Interf}, \eqref{Eq:epi}, and  \eqref{Eq:one_over_k_ana_sq} gives
    \begin{eqnarray}
        \frac{\tilde{\sigma}^2}{\eta}\mathsf{E}\left.\left[\frac{(\mathbf{I}_0^{(n)})^2}{K^2}\right|K>0\right] \leqslant \frac{8P \tilde{\sigma}^2}{\eta p_{\text{s}}(\alpha-2)MR^{\alpha}}\cdot\left(\mathsf{E}\left.\left[\frac{1}{K}\right|K>0\right]-\frac{\bar{K}'e^{-\bar{K}'}}{1-e^{-\bar{K}'}}\right). 
        \label{Eq:pert_lemma}
    \end{eqnarray}
On the other hand, based on truncated channel inversion in \eqref{Eq:power_control},  the expected transmission power of a device is obtained as 
\begin{eqnarray}
    \mathsf{E}[P_X] & = & \mathsf{E}\left[\frac{\eta}{G_{X}|X|^{-\alpha}}\right]  \nonumber \\
    & = & \eta\mathsf{E}\left[\frac{1}{G_X}\right]\cdot\int_0^R  r^{\alpha}f_R(r) dr  \nonumber \\
    & = & \frac{2\eta R^{\alpha}(-\mathrm{Ei}(-g_{\text{th}}))}{\alpha+2}. 
\end{eqnarray}
Under  the average power constraint $\mathsf{E}[P_X] = P$, 
\begin{eqnarray}\label{Eq:eta}
    \eta = \frac{P(\alpha+2)}{2R^{\alpha}(-\mathrm{Ei}(-g_{\text{th}}))}.
\end{eqnarray}
Substituting  \eqref{Eq:eta} into \eqref{Eq:pert_lemma} gives the desired result.

\bibliographystyle{IEEEtran}

\begin{thebibliography}{10}
\providecommand{\url}[1]{#1}
\csname url@samestyle\endcsname
\providecommand{\newblock}{\relax}
\providecommand{\bibinfo}[2]{#2}
\providecommand{\BIBentrySTDinterwordspacing}{\spaceskip=0pt\relax}
\providecommand{\BIBentryALTinterwordstretchfactor}{4}
\providecommand{\BIBentryALTinterwordspacing}{\spaceskip=\fontdimen2\font plus
\BIBentryALTinterwordstretchfactor\fontdimen3\font minus
  \fontdimen4\font\relax}
\providecommand{\BIBforeignlanguage}[2]{{%
\expandafter\ifx\csname l@#1\endcsname\relax
\typeout{** WARNING: IEEEtran.bst: No hyphenation pattern has been}%
\typeout{** loaded for the language `#1'. Using the pattern for}%
\typeout{** the default language instead.}%
\else
\language=\csname l@#1\endcsname
\fi
#2}}
\providecommand{\BIBdecl}{\relax}
\BIBdecl

\bibitem{8736011}
Z.~{Zhou}, X.~{Chen}, E.~{Li}, L.~{Zeng}, K.~{Luo}, and J.~{Zhang}, ``Edge
  intelligence: Paving the last mile of artificial intelligence with edge
  computing,'' \emph{Proc. of the IEEE}, vol. 107, no.~8, pp. 1738--1762, 2019.

\bibitem{zhu2020toward}
G.~Zhu, D.~Liu, Y.~Du, C.~You, J.~Zhang, and K.~Huang, ``Toward an intelligent
  edge: wireless communication meets machine learning,'' \emph{IEEE Comm.
  Magazine}, vol.~58, no.~1, pp. 19--25, 2020.

\bibitem{lim2020federated}
W.~Y.~B. Lim, N.~C. Luong, D.~T. Hoang, Y.~Jiao, Y.-C. Liang, Q.~Yang,
  D.~Niyato, and C.~Miao, ``Federated learning in mobile edge networks: A
  comprehensive survey,'' \emph{IEEE Commun. Surveys Tuts.}, 2020.

\bibitem{chen2019joint}
M.~{Chen}, Z.~{Yang}, W.~{Saad}, C.~{Yin}, H.~V. {Poor}, and S.~{Cui}, ``A
  joint learning and communications framework for federated learning over
  wireless networks,'' \emph{IEEE Trans. Wireless Commun.}, vol.~20, no.~1, pp.
  269--283, 2021.

\bibitem{shi2020joint}
W.~{Shi}, S.~{Zhou}, Z.~{Niu}, M.~{Jiang}, and L.~{Geng}, ``Joint device
  scheduling and resource allocation for latency constrained wireless federated
  learning,'' \emph{IEEE Trans. Wireless Commun.}, vol.~20, no.~1, pp.
  453--467, 2021.

\bibitem{ren2020scheduling}
\BIBentryALTinterwordspacing
J.~Ren, Y.~He, D.~Wen, G.~Yu, K.~Huang, and D.~Guo, ``Scheduling in cellular
  federated edge learning with importance and channel awareness,'' 2020.
  [Online]. Available: \url{arXiv:2004.00490v2}
\BIBentrySTDinterwordspacing

\bibitem{yang2020delay}
\BIBentryALTinterwordspacing
Z.~Yang, M.~Chen, W.~Saad, C.~S. Hong, M.~Shikh-Bahaei, H.~V. Poor, and S.~Cui,
  ``Delay minimization for federated learning over wireless communication
  networks,'' 2020. [Online]. Available: \url{https://arxiv.org/abs/2007.03462}
\BIBentrySTDinterwordspacing

\bibitem{song2020optimal}
J.~{Song} and M.~{Kountouris}, ``Optimal number of edge devices in distributed
  learning over wireless channels,'' in \emph{IEEE  Workshop
on  Signal Process. Adv. Wireless Commun. (SPAWC) (virtual)}, May 26-29, 2020.

\bibitem{zhu2020over}
\BIBentryALTinterwordspacing
G.~Zhu, J.~Xu, and K.~Huang, ``Over-the-air computing for wireless data aggregation in massive IoT,'' 2020. [Online]. Available:
  \url{https://arxiv.org/abs/2009.02181}
\BIBentrySTDinterwordspacing

\bibitem{amiri2020machine}
M.~M. Amiri and D.~G{\"u}nd{\"u}z, ``Machine learning at the wireless edge:
  Distributed stochastic gradient descent over-the-air,'' \emph{IEEE Trans.
  Signal Process.}, vol.~68, pp. 2155--2169, 2020.

\bibitem{zhu2020one}
G.~{Zhu}, Y.~{Du}, D.~{Gündüz}, and K.~{Huang}, ``One-bit over-the-air
  aggregation for communication-efficient federated edge learning: Design and
  convergence analysis,'' \emph{to appear in IEEE Trans. Wireless Commun.},
  2020.

\bibitem{zhang2020gradient}
\BIBentryALTinterwordspacing
N.~Zhang and M.~Tao, ``Gradient statistics aware power control for over-the-air
  federated learning in fading channels,'' 2020. [Online]. Available:
  \url{https://arxiv.org/abs/2003.02089}
\BIBentrySTDinterwordspacing

\bibitem{yang2020federated}
K.~Yang, T.~Jiang, Y.~Shi, and Z.~Ding, ``Federated learning via over-the-air
  computation,'' \emph{IEEE Trans. Wireless Commun.}, vol.~19, no.~3, pp.
  2022--2035, 2020.

\bibitem{lin2017deep}
\BIBentryALTinterwordspacing
Y.~Lin, S.~Han, H.~Mao, Y.~Wang, and W.~J. Dally, ``Deep gradient compression:
  Reducing the communication bandwidth for distributed training,'' 2017.
  [Online]. Available: \url{https://arxiv.org/abs/1712.01887}
\BIBentrySTDinterwordspacing

\bibitem{amiri2020federated}
M.~M. Amiri and D.~G{\"u}nd{\"u}z, ``Federated learning over wireless fading
  channels,'' \emph{IEEE Trans. Wireless Commun.}, vol.~19, no.~5, pp.
  3546--3557, 2020.

\bibitem{yang2019scheduling}
H.~H. Yang, Z.~Liu, T.~Q. Quek, and H.~V. Poor, ``Scheduling policies for
  federated learning in wireless networks,'' \emph{IEEE Trans. Commun.},
  vol.~68, no.~1, pp. 317--333, 2019.

\bibitem{haenggi2009stochastic}
M.~Haenggi, J.~G. Andrews, F.~Baccelli, O.~Dousse, and M.~Franceschetti,
  ``Stochastic geometry and random graphs for the analysis and design of
  wireless networks,'' \emph{IEEE J. Sel. Areas Commun.}, vol.~27, no.~7, pp.
  1029--1046, 2009.

\bibitem{kingman2005p}
J.~Kingman, \emph{\BIBforeignlanguage{English}{Poisson Processes}}, ser. Oxford
  Studies in Probability.\hskip 1em plus 0.5em minus 0.4em\relax United
  Kingdom: Oxford University Press, 1993.

\bibitem{chiu2013stochastic}
S.~N. Chiu, D.~Stoyan, W.~S. Kendall, and J.~Mecke, \emph{Stochastic geometry
  and its applications}.\hskip 1em plus 0.5em minus 0.4em\relax John Wiley \&
  Sons, 2013.

\bibitem{andrews2011tractable}
J.~G. Andrews, F.~Baccelli, and R.~K. Ganti, ``A tractable approach to coverage
  and rate in cellular networks,'' \emph{IEEE Trans. Commun.}, vol.~59, no.~11,
  pp. 3122--3134, 2011.

\bibitem{hosseini2016stochastic}
K.~Hosseini, W.~Yu, and R.~S. Adve, ``A stochastic analysis of network {MIMO}
  systems,'' \emph{IEEE Trans. Signal Process.}, vol.~64, no.~16, pp.
  4113--4126, 2016.

\bibitem{dhillon2012modeling}
H.~S. Dhillon, R.~K. Ganti, F.~Baccelli, and J.~G. Andrews, ``Modeling and
  analysis of {K}-tier downlink heterogeneous cellular networks,'' \emph{IEEE
  J. Sel. Areas Commun.}, vol.~30, no.~3, pp. 550--560, 2012.

\bibitem{soh2013energy}
Y.~S. Soh, T.~Q. Quek, M.~Kountouris, and H.~Shin, ``Energy efficient
  heterogeneous cellular networks,'' \emph{IEEE J. Sel. Areas Commun.},
  vol.~31, no.~5, pp. 840--850, 2013.

\bibitem{chetlur2017downlink}
V.~V. Chetlur and H.~S. Dhillon, ``Downlink coverage analysis for a finite
  3-{D} wireless network of unmanned aerial vehicles,'' \emph{IEEE Trans.
  Commun.}, vol.~65, no.~10, pp. 4543--4558, 2017.

\bibitem{1542405}
S.~P. {Weber}, {Xiangying Yang}, J.~G. {Andrews}, and G.~{de Veciana},
  ``Transmission capacity of wireless ad hoc networks with outage
  constraints,'' \emph{IEEE Trans. Inf. Theory}, vol.~51, no.~12, pp.
  4091--4102, 2005.

\bibitem{bernstein2018signsgd}
J.~Bernstein, Y.-X. Wang, K.~Azizzadenesheli, and A.~Anandkumar, ``signsgd:
  Compressed optimisation for non-convex problems,'' in \emph{Int. Conf. Mach.
  Learn. (ICML)}, pp. 560--569, Stockholm, Sweden, 2018.

\bibitem{Rappaport1996Wireless}
T.~S. Rappaport, \emph{Wireless Communications: Principles and Practice},
  1st~ed.\hskip 1em plus 0.5em minus 0.4em\relax IEEE Press, 1996.

\bibitem{zhu2019broadband}
G.~Zhu, Y.~Wang, and K.~Huang, ``Broadband analog aggregation for low-latency
  federated edge learning,'' \emph{IEEE Trans. Wireless Commun.}, vol.~19,
  no.~1, pp. 491--506, 2019.

\bibitem{konevcny2017stochastic}
\BIBentryALTinterwordspacing
J.~Kone{\v{c}}n{\`y}, ``Stochastic, distributed and federated optimization for
  machine learning,'' 2017. [Online]. Available:
  \url{https://arxiv.org/pdf/1707.01155.pdf}
\BIBentrySTDinterwordspacing

\bibitem{mcmahan2017communication}
B.~McMahan, E.~Moore, D.~Ramage, S.~Hampson, and B.~A. y~Arcas,
  ``Communication-efficient learning of deep networks from decentralized
  data,'' in \emph{Proc. Int. Conf. Artif. Int. Statist. (AISTATS)},  pp.
  1273--1282, Ft. Lauderdale, FL, Apr. 20 - 22, 2017.

\bibitem{goldsmith2005wireless}
A.~Goldsmith, \emph{Wireless Communications}.\hskip 1em plus 0.5em minus
  0.4em\relax USA: Cambridge University Press, 2005.

\bibitem{weber2007effect}
S.~Weber, J.~G. Andrews, and N.~Jindal, ``The effect of fading, channel
  inversion, and threshold scheduling on ad hoc networks,'' \emph{IEEE Trans.
  Inf. Theory}, vol.~53, no.~11, pp. 4127--4149, 2007.

\bibitem{cao2020optimized}
X.~Cao, G.~Zhu, J.~Xu, and K.~Huang, ``Optimized power control for over-the-air
  computation in fading channels,'' \emph{IEEE Trans. Commun.}, vol.~19,
  no.~11, pp. 7498--7513, 2020.

\bibitem{andrews2016primer}
\BIBentryALTinterwordspacing
J.~G. Andrews, A.~K. Gupta, and H.~S. Dhillon, ``A primer on cellular network
  analysis using stochastic geometry,'' 2016. [Online]. Available:
  \url{https://arxiv.org/pdf/1604.03183.pdf}
\BIBentrySTDinterwordspacing

\bibitem{chao1972negative}
M.-T. Chao and W.~Strawderman, ``Negative moments of positive random
  variables,'' \emph{J. Amer. Stat. Assoc.}, vol.~67, no. 338, pp. 429--431,
  1972.

\bibitem{stephan1945expected}
F.~F. Stephan, ``The expected value and variance of the reciprocal and other
  negative powers of a positive bernoullian variate,'' \emph{Ann. Math. Stat.},
  vol.~16, no.~1, pp. 50--61, 1945.

\end{thebibliography}

\end{document}